\documentclass[11 pt]{amsart}

\usepackage{amssymb, amsmath, mathrsfs}
\usepackage[margin=3.5 cm]{geometry}
\usepackage{mathabx}
\usepackage{upgreek}
\usepackage{graphicx}
\usepackage{color}
\usepackage{subfigure}
\newtheorem{theorem}{Theorem}[section]
\newtheorem{prop}{Proposition}[section]
\newtheorem{claim}{Claim}[section]
\newtheorem{coro}{Corollary}[section]
\newtheorem{lemma}[theorem]{Lemma}

\newtheorem{fact}{Fact}[section]

\theoremstyle{definition}
\newtheorem{definition}[theorem]{Definition}

\theoremstyle{remark}
\newtheorem{remark}[theorem]{Remark}

\numberwithin{equation}{section}

\newcommand{\abs}[1]{\lvert#1\rvert}
\newcommand{\norm}[1]{\lVert#1\rVert}

\DeclareMathOperator{\im}{Im}
\newcommand{\ud}{\mathrm{d}}

\begin{document}

\title[Spectral data from rational frequency approximants]{Analytic quasi-periodic Schr\"odinger operators and rational frequency approximants}

\author{S. Jitomirskaya and C. A. Marx}
\address{Department of Mathematics, University of California, Irvine CA, 92717}
\thanks{The work was supported by NSF Grant DMS-1101578.}





\begin{abstract}
Consider a quasi-periodic Schr\"odinger operator  $H_{\alpha,\theta}$ with analytic potential and irrational frequency $\alpha$. Given any rational approximating $\alpha$, let $S_+$ and $S_-$ denote the union, respectively, the intersection of the spectra taken over $\theta$. We show that up to sets of zero Lebesgue measure, the absolutely continuous spectrum can be obtained asymptotically from $S_-$ of the periodic operators associated with the continued fraction expansion of $\alpha$. This proves a conjecture of Y. Last in the analytic case. Similarly, from the asymptotics of $S_+$, one recovers the spectrum of $H_{\alpha,\theta}.$ 
\end{abstract}

\maketitle

\section{introduction}

Let $\mathbb{T}:=\mathbb{R}/\mathbb{Z}$. Consider the quasi-periodic
Schr\"odinger operator with potential generated from an analytic function $v \in \mathcal{C}^\omega(\mathbb{T},\mathbb{R})$,
\begin{eqnarray} \label{eq_defn_op}
H_{\alpha, \theta}: \mathit{l}^2(\mathbb{Z}) \to \mathit{l}^2(\mathbb{Z}) ~\mbox{,}
~(H_{\alpha, \theta} \psi)_n := \psi_{n-1} + \psi_{n+1} + v(\alpha n + \theta) \psi_n ~\mbox{.}
\end{eqnarray}
Here, $\alpha \in \mathbb{T}$ is a fixed irrational usually referred to as frequency and $\theta \in \mathbb{T}$ is the phase. For fixed $\theta$, denote by $\sigma(\alpha, \theta)$ and $\sigma_{\mathrm{ac}}(\alpha, \theta)$ the spectrum of $H_{\alpha,\theta}$ and its absolutely continuous (ac)-component, respectively. It is well known that since $\alpha$ is irrational, the spectrum and ac spectrum do not depend on $\theta$ \cite{pastur,ls1}:
\begin{eqnarray}
\sigma(\alpha, \theta) & 
=:&\Sigma(\alpha) ~\mbox{,} \label{eq_spectrum} \\
\sigma_{\mathrm{ac}}(\alpha, \theta) &
=:&\Sigma_{\mathrm{ac}}(\alpha) ~\mbox{,} ~\forall \theta
\in \mathbb{T} ~\mbox{.} \label{eq_acspectrum}
\end{eqnarray}

Operators of the form (\ref{eq_defn_op}) arise as effective Hamiltonians in the description of a crystal layer immersed in an external magnetic field. In this application, $\alpha$ represents the magnetic flux through a unit cell and $v$ contains information about the lattice geometry as well as the interaction between lattice sites. Such operators, beginning with their prototype,  
the almost Mathieu (or Harper's) operator where 
%
$v(x) = 2 \lambda \cos(2 \pi x) ~\mbox{, } \lambda \in \mathbb{R} ~\mbox{,}$
have been subject of extensive rigorous, heuristic and numerical studies. The latter, naturally, always deal only with rational frequencies $p_n/q_n$ approximating $\alpha,$ with conclusions then made about the irrational case. For example, the famous Hofstadter butterfly \cite{hof} is a plot of the almost Mathieu spectra for 50 rational values of $\alpha.$ It is therefore an important and  natural question if and in what sense the spectral properties of such rational approximants relate to those of the quasi-periodic operator $H_{\alpha, \theta}$. The purpose of this article is to show that spectrum and ac spectrum of $H_{\alpha, \theta}$ can be associated with natural limits of spectra of the approximants, in a rather strong sense.

The basic spectral properties of operators associated with rational values of $\alpha$, $(\frac{p}{q})
$ with $(p,q)=1$ 
are well understood. For each $\theta \in \mathbb{T}$, $H_{\frac{p}{q}, \theta}$ is a  {\em{periodic}} operator whose spectrum, $\sigma(\frac{p}{q}, \theta)$, is given in terms of the discriminant, $t_{\frac{p}{q}}(\theta, E)$,
\begin{equation} \label{eq_spectrumper}
\sigma(\frac{p}{q}, \theta) = t_{\frac{p}{q}}(\theta, .)^{-1} ( [-2,2] ) ~\mbox{,} 
\end{equation}
where
\begin{eqnarray}
t_{\frac{p}{q}}(\theta, E) := \mathrm{tr} \left\{ \prod_{j=q - 1}^{0} A^E\left(j \frac{p}{q} + \theta \right) \right\} ~\mbox{,} \\
~A^E(x) = \begin{pmatrix} E- v(x) & -1 \\ 1 & 0  \end{pmatrix} ~\mbox{.} \label{eq_transferm}
\end{eqnarray}
Standard arguments show that $\sigma(\frac{p}{q}, \theta)$ is purely absolutely continuous and consists of $q$, possibly touching, bands (see also Fact \ref{fact_floquet} below).

In order to treat rational and irrational frequencies on the same footing, similar to Avron, v. Mouche, and Simon \cite{A}, given $\beta \in \mathbb{T}$, we introduce the sets
\begin{eqnarray}
S_{+}\left(\beta \right) := \bigcup_{\theta \in \mathbb{T}} \sigma\left(\beta,\theta\right) ~\mbox{,}  \\
S_{-}\left(\beta \right) := \bigcap_{\theta \in \mathbb{T}} \sigma_{\mathrm{ac}} \left(\beta,\theta\right) ~\mbox{.}
\end{eqnarray}
In order to avoid confusion, $\beta$ will always denote an arbitrary rational {\em{or}} irrational element of $\mathbb{T}$, whereas $\alpha \in \mathbb{T}$ is reserved for an irrational. 
From (\ref{eq_spectrum}) and (\ref{eq_acspectrum}), we infer that $S_+(\alpha)=\Sigma(\alpha)$ and $S_-(\alpha) = \Sigma_{\mathrm{ac}}(\alpha)$.

Given $\beta \in \mathbb{T}$, the set $S_+(\beta)$ has a neat interpretation as the spectrum of the decomposable operator
\begin{equation} 
H_\beta^\prime:=\int_\mathbb{T}^\oplus H_{\beta,\theta} \ud \theta ~\mbox{,}
\end{equation}
acting on the space $\int_\mathbb{T}^\oplus \mathit{l}^2(\mathbb{Z}) \ud \theta$. A proof of this simple, but useful fact is given in Proposition \ref{prop_ints+}, Sec. \ref{sec_genduality}. In particular, this illustrates that for any $\beta$, not necessarily irrational, $S_+(\beta)$ is really the natural quantity in the study of the spectrum of the family of operators $\{H_{\beta,\theta}\}_{\theta \in \mathbb{T}}$. 

In this article we analyze the continuity of the sets $S_\pm$ upon rational approximation of $\alpha$. To this end, let $\mathscr{B}$ denote the quotient space of the Borel sets of $\mathbb{R}$ modulo sets of zero Lebesgue measure. For convenience we will suppress the distinction between an equivalence class and its representatives. Given $(B_n)_{n\in \mathbb{N}}$  and $B, $ Borel subsets of $\mathbb{R},$ we write
\begin{eqnarray} \label{eq_def_topology}
\lim_{n\to \infty} B_n  = B ~\mbox{(in $\mathscr{B}$)} ~:\Longleftrightarrow & 
\limsup_{n \to \infty} B_n = \liminf_{n \to \infty} B_n = B 
\nonumber \\
\Longleftrightarrow & \lim_{n \to \infty} \chi_{B_n} 
 = \chi_B ~\mbox{Lebesgue a.e. ,}
\end{eqnarray}
which induces a topology on $\mathscr{B}$. 
Here, 
as usual, 
\begin{equation}
\liminf_{n \to \infty} B_n :=\bigcup_{n \in \mathbb{N}} \bigcap_{k \geq n} B_k ~\mbox{,} ~ \limsup_{n \to \infty} B_n :=\bigcap_{n \in \mathbb{N}} \bigcup_{k \geq n} B_k ~\mbox{.}
\end{equation}
Following, all set limits are understood in the topology given in (\ref{eq_def_topology}). Also, if not mentioned explicitly, all relations between Borel sets are understood as relations of the associated equivalence classes in $\mathscr{B}$.

Trivially, 
\begin{equation} \label{eq_trivial}
B_n \to B \Rightarrow \vert B_n \vert \to \vert B \vert ~\mbox{,}
\end{equation}
where $\vert . \vert$ is the Lebesgue measure.\\
%


Our main result recovers the sets $S_\pm(\alpha)$ from the asymptotics of $S_\pm(p_n/q_n)$ for a sequence of convergents $\frac{p_n}{q_n}\to\alpha:$
\begin{theorem} \label{thm_main}
For any irrational $\alpha$ there is a sequence $\frac{p_n}{q_n}\to\alpha$ such that
\begin{itemize}
\item[(i)] $\lim_{n\to\infty} S_-\left(\frac{p_n}{q_n}\right) = S_-(\alpha)  =  \Sigma_{\mathrm{ac}}\left(\alpha\right)$,
\item[(ii)] $\lim_{n\to\infty} S_+\left(\frac{p_n}{q_n}\right) = S_+(\alpha)  =  \Sigma\left(\alpha\right)$.
\end{itemize}
\end{theorem}
In particular, from (\ref{eq_trivial}) we obtain as an immediate corollary
\begin{coro} \label{coro_mainthm}
For $\alpha$ and $\left(\frac{p_n}{q_n}\right)$ as in Theorem \ref{thm_main},
\begin{eqnarray}
\lim_{p_n/q_n \to \alpha} \vert S_+\left(\frac{p_n}{q_n}\right) \vert = \vert \Sigma(\alpha) \vert ~\mbox{,}  \label{eq_meas_coro_1} ~ \\
\lim_{p_n/q_n \to \alpha} \vert S_-\left(\frac{p_n}{q_n}\right) \vert = \vert \Sigma_{\mathrm{ac}}(\alpha) \vert ~\mbox{.} \label{eq_meas_coro_2}
\end{eqnarray}
\end{coro}

\begin{remark}
\begin{enumerate}
\item The fact that the limits exist is a part of the statement of the theorem.
\item Theorem \ref{thm_main} is new even for the almost Mathieu operator. The new part here is in establishing statements about the $\liminf$ which show, in some sense, that there are eventually  no traveling gaps.
\item From a practical point of view, mere existence of {\em{some}} sequence along which $S_\pm$ of the quasi-periodic operator can be reconstructed from its periodic approximants is enough, as computations for $S_\pm$ can usually be done for an arbitrary rational (see e.g. \cite{A}, for the almost Mathieu operator). Nevertheless, we can give the following explicit characterization of $(p_n/q_n)$ in Theorem \ref{thm_main}: For Diophantine $\alpha$ (see (\ref{eq_diophdef})) the sequence $\left(\frac{p_n}{q_n}\right)$ can be taken as the sequence of canonical continuous fraction approximants (see (\ref{eq_cf1}). For the non-Diophantine case, for our proof, we have to restrict to a subsequence of sufficiently strong approximants (see (\ref{eq_liouv}) for details). 
We mention that based on a preprint of this article, Artur Avila has pointed out to us that Theorem \ref{thm_main} may however be strengthened so that the approximating sequence is the full sequence of canonical continuous fractions approximants even for non-Diophantine $\alpha$ (see Sec. \ref{subsec_avila} for details).
\item 
Analyticity of $v$ is essential for {\it our proof} of Theorem
\ref{thm_main}, as it allows for a generalization of Chambers' formula
(see Proposition \ref{prop_chambers}). We don't believe however
analyticity should be essential for Theorem \ref{thm_main}, see more below. 
\item $|S_+(\beta)|$ and therefore $S_+(\beta)$ can be discontinuous at rationals (see, e.g., Fact 1, Sec. 7 in \cite{A}). We don't have such evidence, however, for $S_-.$
Indeed, for the almost Mathieu operator, $|S_-(\beta)|=|4-2|\lambda||=\mbox{const}.$ It is therefore an interesting question what in general is true in this regard.
\item There is a tempting analogy between the statement of Theorem \ref{thm_main} and the characterization of essential and ac spectra 
through those of the right limits \cite{ls, ls1, remling}. There is no direct relation, though, and the proofs are completely different.
\item The more difficult and interesting part of Theorem \ref{thm_main} is the statement about $S_-$. The argument for part (ii) is close to a proper subset of the argument for part (i).
\end{enumerate}
\end{remark}

Questions about continuity of the sets $S_\pm$ w.r.t. $\beta$ have attracted much attention in the literature, particularly in context of the almost Mathieu operator, where, based on Chambers' formula and symmetry, some computations can be done explicitly. Most of these known results addressed continuity from the weaker point of view of Corollary \ref{coro_mainthm} and were motivated by the Aubry-Andre conjecture \cite{aa} on the measure of the almost Mathieu spectrum, popularized by B. Simon \cite{Simon_1984,Simon_2000}. 
Eq. (\ref{eq_meas_coro_1}) was first obtained for the almost Mathieu operator for a.e. $\alpha$ \cite{Last2, Last1, A} based on 1/2-H\"older continuity of $S_+(\beta)$ in the Hausdorff metric \cite{A}. It was extended to all irrational $\alpha$ by a combination of \cite{JK} and \cite{AK}. The related results of \cite{JK,Last1} hold for all analytic $v,$\footnote{The result of \cite{Last1} is formulated for the almost Mathieu only but the proof holds for any Lipschitz $v$. See more in Sec. \ref{subsec_histrem}.} as a result, for analytic potentials Eq. (\ref{eq_meas_coro_1})  was known for $\alpha$ with unbounded coefficients in the continued fraction expansion (a full measure set) and for all irrational $\alpha$ in the regime of positive Lyapunov exponents. Similarly, \cite{Last2} essentially contains an inequality ($\leq$) in Eq. (\ref{eq_meas_coro_2}) with the limit in (\ref{eq_meas_coro_2}) replaced by $\limsup$, but for a.e. $\alpha.$
The actual Eq. (\ref{eq_meas_coro_2}) was known for the almost Mathieu operator only (\cite{Last1,A} for a.e. $\alpha,\lambda$ and \cite{Last2,JK,AK} extending to all).

Theorem \ref{thm_main} (i) was roughly conjectured by Y. Last, who
informed us that he can establish a variant of this theorem where $v$
is merely $\mathcal{C}^1$ (rather than analytic), but then the
statement only covers a dense $G_\delta$ set of irrational $\alpha$
and appropriate sequences of rationals that approximate these $\alpha$
sufficiently well. Whether or not the analyticity requirement of $v$
can be relaxed without 
reducing the range of frequencies for which the statement holds is an interesting open problem.

In \cite{F}, M. Shamis obtained an inclusion,
\begin{eqnarray}
 \limsup_{\frac{p}{q} \to \alpha} S_{-}\left( \frac{p}{q} \right) \subseteq & \Sigma_{\mathrm{ac}}(\alpha) 
 ~\mbox{,} \label{eq_shamis}
\end{eqnarray}
as a corollary of continuity of the Lyapunov exponent \cite{C}, for all irrational $\alpha$ and arbitrary sequences of rational approximants, which, of course, immediately implies an inequality ($\leq$) in (\ref{eq_meas_coro_2}), for arbitrary  approximants, where the limit in (\ref{eq_meas_coro_2}) is replaced by $\limsup$. \footnote{Shamis also obtained that $\Sigma_{\mathrm{ac}} $ is contained in the $\liminf$ of intersections of Hausdorff $\epsilon-$neighborhoods of $\sigma\left( \frac{p}{q}, \theta \right)$.} 
 The $\limsup$ part of Theorem \ref{thm_main} (ii) is immediate as a corollary of Hausdorff continuity. As mentioned, the $\liminf$ part of Theorem \ref{thm_main} was not known in any setting and is the main subject of this work.


A more detailed review of the history of continuity statements of $S_\pm$ is given below in Sec. \ref{subsec_histrem}.

An important ingredient for our proof is that ac spectrum implies exponentially small variation (in $q$) of the approximating discriminants, obtained as a corollary to the proof of Avila's quantization of acceleration \cite{B} (``{\em{generalized Chambers' formula}}''). This essentially reduces the argument to showing that the phase-averaged discriminants are not only growing not more than sub-exponentially in the denominator, but eventually belong to $[-2,2]$ for a.e. energy in $\Sigma_{\mathrm{ac}}(\alpha)$. This is achieved, in part, through estimates on the level-sets of discriminants (see Sec. 6), which may also be of independent interest.  Altogether our arguments imply a formulation of Theorem \ref{thm_main} (i), in terms of the discriminants of periodic approximants, which is given in Theorem \ref{thm_mainindividphase}.

%
%

We structure the paper as follows. Section \ref{sec_prelim} serves as a preliminary introducing some basic notions. For Diophantine $\alpha$, we first argue that based on \cite{H, AK}, it is enough to consider energies $E$ for which the cocycle $(\alpha, A^E)$ is reducible to constant rotations (see Defintion \ref{def_conjtorot}). In Sec. \ref{sec_chambers}, we then prove above mentioned generalization of the celebrated Chambers' formula to arbitrary analytic potentials, which will allow the analysis of $S_\pm$. This result, formulated in Theorem \ref{prop_chambers}, is based on Avila's 
proof of quantization of the acceleration (\cite{B}; see also Appendix \ref{app_2} in the present paper). 

Section \ref{sec_outline} reduces the further analysis to three cases, two of which are non-trivial and are the subject of the sections \ref{sec_dualtiy} and \ref{sec_case3}. On the way, some general measure estimates for the sub-level sets of real polynomials, whose number of distinct {\em{real}} roots equals their degree, will be needed (Theorem \ref{thm_sublevel_rII}). These considerations are given in Sec. \ref{sec_polyn}. In this context, we also prove that the contributions of individual bands to the level sets are extremized by Chebyshev polynomials of the first kind (see Theorem \ref{thm_poly}), which extends a result of \cite{D}. Even though the latter is not needed for the proof of the main results, we believe Theorem \ref{thm_poly} to be of independent interest, whence include its proof in Sec. \ref{subsec_cheby}.

Combining the pieces, in Sec. \ref{sec_case3} we prove Theorem \ref{thm_main}. Finally, in Sec. \ref{sec_genduality}, some general facts on duality for arbitrary continuous potentials are presented, extending some known results for the almost Mathieu operator. 

\subsection{An alternative argument with an improvement} \label{subsec_avila}
Based on a preprint of this article, Artur Avila has pointed out to us an idea of an alternative proof of the ``intersection spectrum conjecture'', Theorem \ref{thm_main} (i), which yields the result for the full sequence of canonical approximants even if $\alpha$ is neither Diophantine nor Liouville. We present a sketch of his argument below. 

Fix $\alpha \in \mathbb{T}$ and let $(p_n/q_n)$ be the (full) sequence of approximants in a continued fraction expansion of $\alpha$. Define $\mathcal{K} \subseteq \mathbb{T}$ as the set of all $\kappa \in [0,1)$ so that 
\begin{equation}
\inf_{p \in \mathbb{Z}} \vert q_n \kappa - p \vert > \frac{1}{n^2} ~\mbox{, eventually.}
\end{equation}
A simple Borel Cantelli argument shows, $\vert K \vert =1$.

For $\beta \in \mathbb{T}$, denote by $N(\beta, E)$ the integrated density of states (IDS) (see, e.g. \cite{cfks} for the definition). Let $E$ such that $N(\alpha, E) \in \mathcal{K}$. Using \cite{H} 
this can be done for a.e. $E \in \Sigma_{\mathrm{ac}}(\alpha)$. An argument adapted from \cite{AJ2}, 
shows that $E \in \mathcal{R}_\alpha$ (for a definition of $\mathcal{R}_\alpha$, see (\ref{eq_defral})) implies (local) Lipschitz continuity of the IDS in the frequency, i.e. for some $\Gamma=\Gamma(E)$
\begin{equation}
\abs{N(\alpha,E) - N(p_n/q_n)} < \frac{\Gamma}{q_n^2} ~\mbox{.}
\end{equation}

In summary, using the definition of $\mathcal{K}$, we thus conclude that eventually,
\begin{equation}
p - 1 + \frac{1}{2 n^2} < q_n N(\frac{p_n}{q_n}, E) < p - \frac{1}{2 n^2} ~\mbox{,}
\end{equation}
for some $1 \leq p \leq q_n$. By Proposition \ref{prop_chambers} below, the discriminant (see (\ref{eq_discqper}) for a definition) eventually exhibits only exponentially small variation with $\theta$, whence exploring a relation between the IDS and the phase averaged discriminant (see e.g. \cite{AD}
) yields that $E$ must eventually be in the intersection of the $p$th bands.

\subsection{Further historical remarks} \label{subsec_histrem}

The relation between the spectral data of almost periodic operators and those of periodic approximants has enjoyed considerable attention in the literature, in particular in the study of the almost Mathieu operator. In addition to what has been said above, we give a more detailed account of related  results.

An important ingredient for statements of the form of Corollary \ref{coro_mainthm} is a modulus of continuity for $S_+$ with respect to the Hausdorff metric. 
We note that the map $\beta \to S_+(\beta)$ is known to be continuous in Hausdorff metric for continuous potentials \cite{AVS}. In \cite{A}, H\"older $1/2$-continuity for $S_+$ was established for any $\mathcal{C}^1$ potential $v(x)$.\footnote{Actually, Lipschitz continuity of $v(x)$ is enough for the proof given in \cite{A}} In the context of the almost Mathieu operator, this was employed in \cite{A, Thou, Last1}  to obtain statements about the measure and the Hausdorff dimension of the spectrum. 

The arguments in \cite{Last1} are easily seen to hold for a general Lipschitz potential, implying upper-semicontinuity of the map $\beta \to \vert S_+(\beta) \vert$ for {\em{all}} $\beta \in \mathbb{T}$. In terms of lower limits, in the same article, Last moreover showed that for the set of irrational $\alpha$ with unbounded elements in their continued fraction expansion, one has
\begin{equation} \label{eq_lscs+}
\vert S_+(\alpha) \vert \leq \liminf_{n \to \infty} \vert S_+(\frac{p_n}{q_n}) \vert ~\mbox{.}
\end{equation}

The restriction to a.e. $\alpha$ in (\ref{eq_lscs+}) is a consequence of only $1/2$-H\"older continuity of $S_+$ in Hausdorff metric. The modulus of continuity however improves to almost Lipschitz on the set of energies with positive Lyapunov exponent \cite{JK}. This fact was originally proven for analytic $v(x)$ \cite{JK}. Recently, (\ref{eq_lscs+}) has been established for rougher potentials $v(x)$ \cite{JMavi}. As we shall make use of the result for analytic $v(x)$ in the present article, we give a detailed statement in Theorem \ref{thm_jk}. In particular, this implies that for all irrational $\alpha$,
\begin{equation} \label{eq_jk}
\left\vert  S_+(\alpha) \cap \{E: L(\alpha,E) > 0\}  \right\vert = \lim_{n \to \infty} \vert S_+(\frac{p_n}{q_n}) \cap \{E: L(\alpha,E) > 0\} \vert ~\mbox{,}
\end{equation}
where $L(\alpha, E)$ is the Lyapunov exponent \cite{JK}.

Based on dynamical systems considerations, which will also play a crucial role in the present article (see Remark \ref{rem_conjtorot}), Avila and Krikorian \cite{AKpreprint} announced and sketched some arguments for
 joint continuity of the maps $(v,\alpha) \to \vert S_+(\alpha) \vert$ and $(v,\alpha) \to \vert S_-(\alpha) \vert$ in $\mathcal{C}^\omega \times DC(\kappa,r)$. Here, $DC(\kappa,r)$ denotes all Diophantines satisfying (\ref{eq_diophdef}) for fixed constants $\kappa, r$.

Finally, we mention that for general {\em{ergodic}} discrete Schr\"odinger operators, the relation between the ac-spectrum and the spectra of certain periodic approximants has been examined by Last in  \cite{Last2, Last3}.

{\bf Acknowledgement:} We are grateful to M. Shamis and S. Sodin for discussions that inspired our work on this subject. We also would like to thank Yoram Last for useful discussions on the history of the subject of this paper. Finally, we are grateful to Artur Avila (see Sec. \ref{subsec_avila}) for useful discussions and his comments following a first preprint of this article. In this respect, we also thank Qi Zhou (see footnote \ref{footnote} following Theorem \ref{thm_main2}) for his remarks.

\section{Preliminaries} \label{sec_prelim}

Throughout the paper, we shall consider analytic, vector-valued functions on $\mathbb{T}$. Given a Banach space $(\mathfrak{X},\norm{.})$, we shall view the analytic $\mathfrak{X}$-valued functions on $\mathbb{T}$ with holomorphic extension to a neighborhood of $\abs{\im{z}} \leq \delta$, $\delta>0$, as a Banach space in its own right equipped with the norm $\norm{X}_\delta:=\sup_{\abs{\im{z}} \leq \delta} \norm{X(z)}$. For our purposes, $\mathfrak{X}$ is either $\mathbb{C}$ or $M_2(\mathbb{C})$.

We start with some preliminaries related to the dynamical properties of solutions to the second order difference equation, 
\begin{equation} \label{eq_seq}
H_{\beta, \theta} \psi = E \psi ~\mbox{,}
\end{equation}
solved over $\mathbb{C}^\mathbb{Z}$. Here, $\beta \in \mathbb{T}$ and $E \in \mathbb{R}$ are fixed. 

Let $A^E(.): \mathbb{T} \to SL(2,\mathbb{C}))$ as defined  in (\ref{eq_transferm}). We call the pair $(\beta, A^E)$ an (analytic) Schr\"odinger -cocycle and understand it as a linear skew-map on $\mathbb{T} \times \mathbb{C}^2$, i.e. 
\begin{equation}
(\beta, A^E)(\theta, v) := (\theta + \beta, A^E(\theta) v) ~\mbox{,} ~\theta \in \mathbb{T}, ~v \in \mathbb{C}^2 ~\mbox{.}
\end{equation}
Iteration of $(\beta, A^E)$ produces the solution of (\ref{eq_seq}) in the sense that
\begin{equation}
(\beta, A^E(\theta))^n \begin{pmatrix} \psi_0 \\ \psi_{-1} \end{pmatrix} = \begin{pmatrix} \psi_{n} \\ \psi_{n-1} \end{pmatrix} ~\mbox{,} ~n \in \mathbb{N} ~\mbox{.}
\end{equation}

The dynamical properties of the cocycle $(\beta, A^E)$ are characterized in terms of the Lyapunov exponent (LE), defined by
\begin{equation} \label{eq_defle}
L(\beta, A^E) := \inf_{n \in \mathbb{N}} \frac{1}{n} \int_{\mathbb{T}} \log \Vert \prod_{j=n-1}^{0} A^E(x + j \beta) \Vert \ud x ~\mbox{.}
\end{equation}
In \ref{eq_defle}, $\norm{.}$ denotes {\em{any}} matrix norm.

Kingman's sub-additive ergodic theorem implies
\begin{equation}
L(\beta, A^E) = \lim_{n \to \infty} \frac{1}{n} \log  \Vert \prod_{j=n-1}^{0} A^E(x + j \beta) \Vert ~\mbox{,}
\end{equation}
for a.e. $x \in \mathbb{T}$, if $\beta$ is irrational, whereas
\begin{equation}
L(\beta, A^E) = \frac{1}{q} \int_{\mathbb{T}} \log \rho\left(\prod_{j=q-1}^{0} A^E(x + j \frac{p}{q}) \right) \ud x ~\mbox{,}
\end{equation}
if $\beta = \frac{p}{q}$ is rational with $(p,q) = 1$. Here, $\rho(M)$ denotes the spectral radius of a given matrix $M \in M_2(\mathbb{C})$. 

In terms of the LE, for any irrational $\alpha$, the set $S_{-}(\alpha)$ is characterized by
\begin{equation} \label{eq_s-le}
S_-(\alpha) = \{E : L(\alpha, A^E) = 0 \} ~\mbox{.}
\end{equation}
Note that (\ref{eq_s-le}) relies on continuity of the LE w.r.t. the energy, which is known for analytic potentials \cite{C}. 


Theorem \ref{thm_main} (i) will follow from upper-semicontinuity of $S_-(\alpha)$ upon rational approximation \cite{F} (the set inclusion in (\ref{eq_shamis})) by establishing 
\begin{equation} \label{eq_theorem}
\{E : L(\alpha, A^E) = 0 \} \subseteq \liminf_{p_n/q_n \to \alpha} S_{-}\left(\frac{p_n}{q_n}\right) ~\mbox{.}
\end{equation}

As mentioned earlier, part (ii) of Theorem \ref{thm_main} is essentially implied by part (i), whence until Sec. \ref{sec_case3} we will focus on the set $S_-$.

Finally, we will need to distinguish Diophantine and non-Diophantine $\alpha$. For each $\alpha$ one can associate the sequence of canonical continued fraction approximants $p_n/q_n$ satisfying
\begin{equation} \label{eq_cf1}
|\alpha -\frac{p_n}{q_n}|<\frac{1}{q_nq_{n+1}}
\end{equation} 
and 
\begin{equation} \label{eq_cf2}
q_{n+1} \geq 2^{n/2} ~\mbox{.}
\end{equation}

We will say that $\alpha$ is {\em{Diophantine}} if 
\begin{equation} \label{eq_diophdef}
\abs{\sin( 2 \pi \alpha n)} > \frac{\kappa}{\abs{n}^r} ~\mbox{,} ~\forall \mathbb{Z}\setminus\{0\} ~\mbox{,}
\end{equation}
for some $0<\kappa < \infty$ and $r > 1$, both, in general, depending on $\alpha$. Equations (\ref{eq_cf1}) and (\ref{eq_diophdef}) imply that for $\alpha$ Diophantine,
\begin{equation} \label{eq_cf3}
q_{n+1} \leq \frac{2 \pi}{\kappa} q_n^r ~\mbox{,}
\end{equation}
where $\kappa$ and $r$ are the same as in  (\ref{eq_diophdef}).

In particular, if $\alpha$ is non-Diophantine then $\forall r>1$, 
\begin{equation} \label{eq_liouv}
q_{n+1} > q_n^r ~\mbox{, infinitely often.}
\end{equation}
For non-Diophantine $\alpha$, it is precisely these subsequences for which Theorem \ref{thm_main} hold. 

If not specified otherwise, to simplify notation, we agree on the following convention: For non-Diophantine $\alpha$, $(p_n/q_n)$ shall always stand for any fixed sub-sequence of the canonical continued fraction approximants satisfying (\ref{eq_liouv}) for some $r$. In the Diophantine case, however, $(p_n/q_n)$ will denote the (full) sequence of canonical continued fraction approximants.

\section{Chambers' formula revisited} \label{sec_chambers}

In order to make statements about the sets $S_\pm$, some information about the phase-dependence of the discriminant is necessary. First, we recall that given $p/q \in \mathbb{Q}$ with $(p,q)=1$, the discriminant is a $q$-periodic function, whence one may write
\begin{equation} \label{eq_discqper}
t_{p/q}(\theta,E) = \sum_{k \in \mathbb{Z}} a_{q,k}(E) \mathrm{e}^{2 \pi i q k \theta} ~\mbox{.}
\end{equation}

For the almost Mathieu operator, the potential $v$ is in fact a trigonometric polynomial of degree 1. Thus in (\ref{eq_discqper}) only the Fourier coefficients with $k=0, \pm q$ survive, resulting in the celebrated Chambers' formula \cite{Chamb} (for a proof see also \cite{BelS}),
\begin{equation} \label{eq_chambersorig}
t_{p/q}(\theta,E) = a_{q,0}(E) - 2 \lambda^q \cos(2 \pi q \theta) ~\mbox{,}
\end{equation}
which gives rise to explicit expressions for $S_\pm$ in terms of the $q$th degree polynomial $a_{q,0}(E)$ \cite{A}. In particular, it shows that phase variations of the discriminant for the sub-critical almost Mathieu operator are exponentially small in $q$.

For arbitrary analytic $v$ and $E\in\Sigma_{ac}(\alpha)$, the following proposition is therefore a generalization of Chambers' formula. It determines $t_{p/q}(\theta,E)$ of rational approximants of $\alpha$, in terms of the phase-average $a_{q,0}(E)$, up to a correction term which is exponentially small in $q$:
\begin{prop} \label{prop_chambers}
There exists a sequence of nested measurable sets $\mathcal{R}_\alpha^{(l)}$, $\Sigma_{\mathrm{ac}}(\alpha) = \cup_{l \in \mathbb{N}} \mathcal{R}_\alpha^{(l)}$, allowing for the following: For each $l \in \mathbb{N}$, $\exists$ $c_l$ such that for $E \in \mathcal{R}_\alpha^{(l)}$ and some $N=N(E) \in \mathbb{N}$ one has:
\begin{equation} \label{eq_chambers}
\sup_{\theta \in \mathbb{T}} \left\vert t_{p/q}(\theta, E) - a_{q,0}(E) \right\vert \leq \mathrm{e}^{- c_l q} ~\mbox{, } 
\end{equation}
whenever $q>N$ and $p/q \in \mathbb{Q}$, $(p,q)=1$.
\end{prop}


In order to define $\mathcal{R}_\alpha^{(l)}$, we will need to consider complexifications of the cocycle in the phase. Since $A^E:\mathbb{T} \to SL(2,\mathbb{C})$ is analytic, for $\epsilon \in \mathbb{R}$, we may consider its complex extension $(\beta, A^E(.+i\epsilon))=:(\beta, A_\epsilon^E(.))$, defined for $\abs{\epsilon} \leq \delta$ and some $\delta>0$. In analogy to (\ref{eq_defle}), we associate the LE of the complexified cocycle $(\beta, A_\epsilon^E)$. It is easy to see that $L(\beta, A_\epsilon^E)$ is a convex, even function of $\epsilon$. 

\begin{definition} \label{def_conjtorot}
Given $\zeta>0$, an analytic Schr\"odinger cocycle $(\beta, A^E)$ is called {\em{$\zeta$-reducible to rotations}}  if 
\begin{equation}
N(\theta+\beta)^{-1} A^E(\theta) N(\theta) = \begin{pmatrix}  \mathrm{e}^{2 \pi i \phi(\theta)} & 0 \\ 0 & \mathrm{e}^{-2 \pi i \phi(\theta)} \end{pmatrix} ~\mbox{,} ~\theta \in \mathbb{T} ~\mbox{,}
\end{equation}
for some $N:\mathbb{T} \to SL(2,\mathbb{C})$ and $\phi:\mathbb{T} \to \mathbb{C}$, with holomorphic extension to a neighborhood of $\abs{\im(z)} \leq \zeta$. 
Moreover, $(\beta, A^E)$ is called {\em{$\zeta$-reducible to constant rotations}} if $\phi(\theta) \equiv \phi_0$, some $\phi_0 \in \mathbb{T}$.
\end{definition}

\begin{remark} \label{rem_conjtorot}
\begin{itemize} 
In \cite{H}, Avila, Fayad and Krikorian prove that for any {\em{irrational}} $\alpha$ and analytic potential $v$
\begin{equation} \label{eq_afk}
\{E: L(\alpha, A^E)=0\} = \bigcup_{\zeta>0} \{E: (\alpha, A^E) ~\mbox{is $\zeta$-reducible to rotations}\} ~\mbox{.}
\end{equation}
In particular, for {\em{Diophantine}} $\alpha$, solution of a cohomological equation thus shows that 
\begin{equation} \label{eq_ak}
\{E: L(\alpha, A^E)=0\} = \bigcup_{\zeta>0} \{E: (\alpha, A^E) ~\mbox{is $\zeta$-reducible to constant rotations}\} ~\mbox{.}
\end{equation}
We mention that, originally, (\ref{eq_ak}) had been obtained in \cite{AK} independently of (\ref{eq_afk}).
\end{itemize}
\end{remark}

We set,
\begin{equation} \label{eq_redpart}
\mathcal{R}_{\alpha} := \bigcup_{l \in \mathbb{N}} \{E: (\alpha, A^E) ~\mbox{is $1/l$-reducible to rotations}\} =: \bigcup_{l \in \mathbb{N}} \mathcal{R}_\alpha^{(l)} ~\mbox{.}
\end{equation}
Using Remark \ref{rem_conjtorot}, for Diophantine $\alpha$,
\begin{equation} \label{eq_defral}
\mathcal{R}_\alpha^{(l)} = \{E: (\alpha, A^E) ~\mbox{is $1/l$-reducible to {\em{constant}} rotations}\} ~\mbox{,}
\end{equation}
(equality holds setwise) for all $l \in \mathbb{N}$. \footnote{By construction of \cite{H}, $N(x)=N^E(x)$ is a measurable function of $E$. The condition that it is analytic in $x$ in a neighborhood of $\abs{\im{z}} \leq \zeta$ is defined by countably many conditions on Fourier coefficients whence measurability of $\mathcal{R}_\alpha^{(l)}$ follows.}

Based on Definition \ref{def_conjtorot}, Remark \ref{rem_conjtorot} and (\ref{eq_shamis}), Theorem \ref{thm_main} is hence implied by:
\begin{theorem} \label{thm_main2}
Given $\alpha$ irrational, 
\begin{equation}
\mathcal{R}_\alpha \subseteq \liminf_{\frac{p_n}{q_n} \to \alpha} S_{-}\left(\frac{p_n}{q_n}\right) ~\mbox{.}
\end{equation}
\end{theorem}
\begin{remark}
\begin{itemize}
\item[(i)] 
By (\ref{eq_redpart}), it suffices to establish that $\forall l \in \mathbb{N}$,
\begin{equation}
\mathcal{R}_\alpha^{(l)} \subseteq \liminf_{\frac{p_n}{q_n} \to \alpha} S_{-}\left(\frac{p_n}{q_n}\right) ~\mbox{.}
\end{equation}
\item[(ii)] In the proof of continuity of $S_-$, it is in fact only the arguments of Sec. \ref{sec_dualtiy} that will discriminate between Diophantine and non-Diophantine $\alpha$. As we will see, the Diophantine case requires more work, which will be based on reducibility to constant rotations. 
\end{itemize}
\end{remark}

The proof of Proposition (\ref{prop_chambers}) is based on the key ingredient of Avila's global theory of one-frequency operators, more specifically on his result stating that $L(\alpha, A_\epsilon^E)$ is a piece-wise linear, convex function with right derivatives in $2 \pi \mathbb{Z}$ (``quantization of acceleration'') \cite{B} \footnote{Based on a first preprint of our article, an alternative proof, not using quantization of acceleration, was pointed out to us by Qi Zhou. It is based on a perturbative argument showing that on $\mathcal{R}_\alpha^{(l)}$, the $q$-step transfer matrices do not grow exponentially in $q$. \label{footnote}}

The following Lemma is a more detailed version of quantization of acceleration, which is implied from Avila's proof in \cite{B}. For the reader's convenience, we present a full proof in Appendix \ref{app_2}, also supplying some more details of Avila's original argument. As standard, given a convex function $f: (a,b) \to \mathbb{R}$, we let $D_+(f)$ denote its right derivative.
\begin{lemma} \label{lemma_avila}
Let $\delta >0$ such that $v(\theta)$ extends holomorphically to a neighborhood of $\abs{\im(z)} \leq \delta$. Given an irrational $\alpha \in \mathbb{T}$ and any $0<\delta_1<\delta$, there exists $K \in \mathbb{N}\cup\{0\}$, $N \in \mathbb{N}$,  and $0<c$ such that $\forall ~E \in \Sigma(\alpha)$:
\begin{equation} \label{eq_chambers3}
\left\vert t_{p/q}(\theta+i\epsilon) - \sum_{\abs{k} \leq K} a_{q,k} \mathrm{e}^{2 \pi i k q (\theta + i \epsilon)} \right\vert \leq \mathrm{e}^{- c q } ~\mbox{,}
\end{equation}
uniformly on $\mathbb{T} \times [-\delta_1, \delta_1]$ and 
\begin{equation} \label{eq_chambers1}
L(\alpha, A_\epsilon^E) = L(p/q, A_\epsilon^E) + o(1) = \frac{1}{q} \log_+ \left( \max_{0 \leq \abs{k} \leq K} \abs{a_{q,k}} \mathrm{e}^{-2\pi \epsilon k q} \right) + o(1) ~\mbox{,}
\end{equation}
uniformly over $0 \leq \abs{\epsilon} \leq \delta_1$, whenever $q > N$, $p/q \in \mathbb{Q}$ with $(p,q)=1$. In particular, the right derivatives of $L(\alpha, A_\epsilon^E)$ w.r.t. $\epsilon$ satisfy
\begin{equation} \label{eq_quantaccel}
D_+(L(\alpha, A_\epsilon^E)) \in 2\pi \{-K,\dots, K\} ~\mbox{.}
\end{equation}
\end{lemma}

\begin{proof}[Proof of Proposition \ref{prop_chambers}]
Let $E \in \mathcal{R}_\alpha^{(l)}$. Referring to (\ref{eq_chambers1}), we set
\begin{equation} \label{eq_chambers2}
M_q(\epsilon) := \max_{0 \leq \abs{k} \leq K} \abs{a_{q,k}} \mathrm{e}^{-2\pi \epsilon k q} =: \abs{a_{q, k_q(\epsilon)}} \mathrm{e}^{-2\pi \epsilon k_q(\epsilon) q} ~\mbox{.}
\end{equation}

Since $A^E(\theta) \in SL(2,\mathbb{R})$, for all $\theta \in \mathbb{T}$, we have $\abs{a_{q,k}} = \abs{a_{q,-k}}$, $k \in \mathbb{Z}$. Thus,
\begin{equation}
M_q(\epsilon) = \max_{0 \leq k \leq K} \abs{a_{q,k}} \mathrm{e}^{2 \pi \abs{\epsilon} k q} = \abs{a_{q,k_q(\epsilon)}} \mathrm{e}^{2 \pi \abs{\epsilon k_q(\epsilon)} q} ~\mbox{,}
\end{equation}
which makes $M_q(\epsilon)$ a symmetric function in $\epsilon$ and shows $\mathrm{sgn}(\epsilon k_q(\epsilon)) \leq 0$.

$\frac{1}{q} \log_+ M_q(\epsilon)$ is a convex function in $\epsilon$, which, by Lemma \ref{lemma_avila}, is uniformly close  to $L(\alpha, A_\epsilon^E)$ on $\abs{\epsilon} \leq \frac{3}{4 l}$ . Thus, differentiability of $L(\alpha, A_\epsilon^E)$ in a neighborhood of $[-\frac{1}{2l}, \frac{1}{2l}]$ implies
\begin{equation}
D_+(\frac{1}{q} \log_+ M_q(\epsilon)) \to D_+(L(\alpha, A_\epsilon^E)) = 0~\mbox{,}
\end{equation}
uniformly in $\epsilon$ on $[-\frac{1}{2l}, \frac{1}{2l}]$ as $\frac{p}{q} \to \alpha$.

Since $\frac{1}{q} \log_+ M_q(\epsilon)$ is piecewise linear with right derivatives in $2 \pi \{-K, \dots, K\}$,  $N$ can be made sufficiently large such that
\begin{equation} \label{eq_deponE}
\frac{1}{q} \log_+ M_q(\epsilon) = \mathrm{const} = : b_q ~\mbox{, } \epsilon \in [-\frac{1}{2l}, \frac{1}{2l}] ~\mbox{,} 
\end{equation}
whenever $q>N$, $p/q \in \mathbb{Q}$ with $(p,q)=1$, and $b_q = o(1)$.

In particular, for suitably chosen $N$, one concludes for $0 \leq k \leq K$,
\begin{equation}
\left\vert a_{q,k} \right\vert \mathrm{e}^{\pi k q/ l}  \leq M_q(\delta) \leq \mathrm{e}^{\frac{1}{2l} \pi q} ~\mbox{,}
\end{equation} 
\begin{equation}
\abs{a_{q,k}} \leq \mathrm{e}^{-\frac{1}{2 l } \pi q} ~\mbox{,} ~1\leq k \leq K ~\mbox{,}
\end{equation}
for all $q>N$.

In summary, on $\mathcal{R}_\alpha^{(l)}$, one may take $K=0$ in (\ref{eq_chambers3}) and (\ref{eq_chambers1}) whenever $q>N$. This implies the claim of the Proposition with $c_l = \frac{1}{2 l } \pi$ .
\end{proof}

We mention that it is through the limit implying (\ref{eq_deponE}), that $N$ in Proposition \ref{prop_chambers} depends on $E$, even though it is derived from Lemma \ref{lemma_avila}, where the respective $N$ is uniform over $\Sigma(\alpha)$. 

\section{Outline of the argument - A tale of three cases} \label{sec_outline}

To begin with, we recall the following well-known facts from Floquet theory (see e.g. \cite{Toda, O})
\begin{fact} \label{fact_floquet}
Let $\frac{p}{q} \in \mathbb{Q}$, $(p,q)=1$. For any $\theta \in \mathbb{T}$ one has:
\begin{itemize}
\item[(i)] $t_\frac{p}{q}(\theta,E)$ is a monic polynomial in $E$ of degree $q$.
\item[(ii)] $t_\frac{p}{q}(\theta,E)$ splits over $\mathbb{R}$ with $q$ distinct roots.
\item[(iii)] $t_\frac{p}{q}(\theta,E)$ is $\geq 2$ at all its local maxima and $\leq -2$ at all its local minima.
\item[(iv)] $\sigma(\frac{p}{q},\theta)=t_\frac{p}{q}(\theta,.)^{-1}([-2,2])$ consists of $q$ bands and is purely ac.
\end{itemize}
\end{fact}
By (\ref{eq_chambers}) it is clear that properties (i) and (ii) are inherited by the phase-average of the discriminant, $a_{q,0}(E)$.

Fix $l \in \mathbb{N}$. Given $E \in \mathcal{R}_\alpha^{(l)}$, we distinguish between three cases:
\begin{description}
\item[Case 1] Eventually,
\begin{equation}
-2 + \mathrm{e}^{ - c_l q_n} < a_{q_n,0}(E) < 2 - \mathrm{e}^{ - c_l q_n} ~\mbox{.}
\end{equation}

\item[Case 2] Infinitely often (i.o.) in $n$,
\begin{equation}
\left\vert a_{q_n,0}(E) \right\vert >  2 + \mathrm{e}^{- c_l q_n} ~\mbox{.}
\end{equation}

\item[Case 3] $\abs{a_{q_n,0}(E) \pm 2} \leq \mathrm{e}^{- c_l q_n}$, i.o. in $n$ .
\end{description}

Trivially,
\begin{equation}
\{E \in \mathcal{R}_\alpha^{(l)}: E ~\mbox{satisfies Case 1}\} \subseteq \liminf_{\frac{p_n}{q_n} \to \alpha} S_-\left( \frac{p_n}{q_n} \right) ~\mbox{.}
\end{equation}
On the other hand, it is also clear that 
\begin{equation}
\{E \in \mathcal{R}_\alpha^{(l)}: E ~\mbox{satisfies Case 2}\} \cap \liminf_{\frac{p_n}{q_n} \to \alpha} S_+\left( \frac{p_n}{q_n} \right) = \emptyset ~\mbox{.}
\end{equation}

The remainder of the paper will thus be devoted to showing that for all $l \in \mathbb{N}$,
\begin{equation} \label{eq_casestoprove}
\left \vert \left\{ E \in \mathcal{R}_\alpha^{(l)}: E ~\mbox{satisfies Case 2 or 3} \right\} \right \vert = 0 ~\mbox{,}
\end{equation}
which by (\ref{eq_redpart}) will prove Theorem \ref{thm_main2}. In the remainder of the paper, we thus let $l \in \mathbb{N}$ be fixed and arbitrary.

\section{Case 2 - Duality} \label{sec_dualtiy}

For non-Diophantine $\alpha$, Case 2 is a straightforward consequence of $1/2$-H\"older continuity of $S_+$ in Hausdorff metric \cite{A}. In the Diophantine case this degree of regularity is insufficient (in fact, $\gamma$-H\"older continuity for any $\gamma > 1/2$ would suffice). The purpose of this section is thus to improve on the degree of regularity of $S_+$ in Hausdorff metric for Diophantine $\alpha$. The easy argument for the non-Diophantine case is given in the end of this section.

We claim,
\begin{prop} \label{prop_case2}
For irrational $\alpha$ and Lebesgue a.e. $E \in \mathcal{R}_\alpha^{(l)}$,  one has
\begin{equation}
\abs{a_{q_n,0}(E)} \leq 2 + \mathrm{e}^{- c_l q_n} ~\mbox{, eventually.}
\end{equation}
\end{prop}

For Diophantine $\alpha$, the proof of Proposition \ref{prop_case2} is based on duality. For $\beta \in \mathbb{T}$, not necessarily irrational, consider the family of operators $\{ H_{\beta, \theta}\}_{\theta \in \mathbb{T}}$ defined in (\ref{eq_defn_op}). We associate its dual, $\{ \hat{H}_{\beta, \xi} \}_{\xi \in \mathbb{T}}$,
\begin{equation} \label{eq_defdual}
\hat{H}_{\beta, \xi} : \mathit{l}^2(\mathbb{Z}) \to \mathit{l}^2(\mathbb{Z}) ~\mbox{,} ~ (\hat{H}_{\beta, \xi} \psi)_n:= (\hat{v} * \psi)_n + 2 \cos(\beta n + \xi) \psi_n ~\mbox{,}
\end{equation}
where $\hat{v}:=(\hat{v}_n)_{n \in \mathbb{Z}}$ is the sequence of Fourier coefficients for $v(\theta)$. Denote the spectrum of $\hat{H}_{\beta, \xi}$ by $\hat{\sigma}(\beta, \xi)$. For irrational $\beta$, ergodicity and minimality of $\theta \mapsto (\theta+\beta)(\mathrm{mod} 1)$ imply the analogue of (\ref{eq_spectrum}).

Following, $\hat{\Sigma}(\beta)$ stands for the spectrum of the ergodic operators $\{ \hat{H}_{\beta, \xi} \}_{\xi \in \mathbb{T}}$. 

A fundamental property of duality is invariance of the set $S_+$:
\begin{theorem} \label{thm_duality_inv}
For Schr\"odinger operators given by (\ref{eq_defn_op}) with continuous potential $v$ and any $\beta \in \mathbb{T}$, we have
\begin{equation} \label{eq_S+dual}
S_+(\beta) = \bigcup_{\xi \in \mathbb{T}} \hat{\sigma}(\beta, \xi) =: \hat{S}_+(\beta) ~\mbox{.} 
\end{equation}
\end{theorem}
Invariance of $S_+$ is known explicitly for the almost Mathieu operator \cite{AVS}. We postpone the proof of Theorem \ref{thm_duality_inv} for general continuous potentials to Sec. \ref{sec_genduality}.

Duality maps $\mathcal{R}_\alpha$ to localized states. To make this precise, we introduce the following terminology:
\begin{definition}
Let $H$ be a bounded self-adjoint operator on $\mathit{l}^2(\mathbb{Z})$. Suppose $E$ is an eigenvalue of $H$. Given $0 < C < \infty$ and $\gamma > 0$, we say $E$ is $(C,\gamma)$-localized for $H$, if for some $\psi \in \ker(H-E)\setminus\{0\}$,
\begin{equation}
\abs{\psi_n} \leq C \mathrm{e}^{ - \gamma \abs{n}} ~\mbox{, } \forall n \in \mathbb{Z} ~\mbox{.}
\end{equation}
\end{definition}

\begin{lemma} \label{lemma_dual}
For $\alpha$ irrational, suppose 
\begin{equation} \label{eq_lemdual}
N(\theta+\alpha)^{-1} A^E(\theta) N(\theta) = R_\theta:=\begin{pmatrix} \mathrm{e}^{2 \pi i \phi_0} & 0 \\ 0 & \mathrm{e}^{-2 \pi i \phi_0} \end{pmatrix} ~\mbox{,} ~\theta \in \mathbb{T} ~\mbox{,}
\end{equation}
for some $\phi_0 \in \mathbb{T}$, analytic $N:\mathbb{T} \to SL(2,\mathbb{C})$ with holomorphic extension to a neighborhood of $\abs{\im(z)} \leq \delta$, some $\delta>0$. Then, $E$ is $(\norm{N}_\delta, 2 \pi \delta)$-localized for $\hat{H}_{\alpha, \theta}$.
\end{lemma}

\begin{proof}
Write 
\begin{equation}
N(\theta) = \begin{pmatrix} a(\theta) & b(\theta) \\ c(\theta) & d(\theta) \end{pmatrix} ~\mbox{.}
\end{equation}
Equation (\ref{eq_lemdual}) implies
\begin{equation}
c(\theta) = \mathrm{e}^{-2 \pi i \phi_0} a(\theta-\alpha) ~\mbox{, } d(\theta) =\mathrm{e}^{2 \pi i \phi_0} b(\theta-\alpha) ~\mbox{.}
\end{equation}
In particular,
\begin{equation} \label{eq_lemdual2}
(E - v(\theta)) a(\theta) - \mathrm{e}^{-2 \pi i \phi_0} a(\theta-\alpha) = \mathrm{e}^{2 \pi i \phi_0} a(\theta+\alpha) ~\mbox{.}
\end{equation}

Taking the Fourier transform of (\ref{eq_lemdual2}) we obtain
\begin{equation}
E \hat{a}_n = (\hat{v} * \hat{a})_n + 2 \cos(2 \pi (\alpha n + \phi_0)) \hat{a}_n ~\mbox{,}
\end{equation}
and since $a(\theta)$ extends holomorphically to $\abs{\im(z)} \leq \delta$,
\begin{equation}
\abs{ \hat{a}_n } \leq \norm{a}_\delta \mathrm{e}^{-2 \pi \abs{n} \delta} \leq \norm{N}_\delta \mathrm{e}^{-2 \pi \abs{n} \delta}  ~\mbox{, } n \in \mathbb{Z} ~\mbox{.}
\end{equation}
\end{proof}

We shall also need
\begin{lemma} \label{lem_conts+}
For $\beta \in \mathbb{T}$, suppose that there exists $E \in \mathbb{R}$ which is $(C,\gamma)$-localized for $\hat{H}_{\beta,\xi}$ and some $\xi \in \mathbb{T}$. Then, for 
all $\beta^\prime \in \mathbb{T}$,
\begin{equation} \label{eq_S+dual1}
\mathrm{dist}(E, S_+(\beta^\prime)) \leq \mathrm{dist}(E, \sigma(\beta^\prime,\xi)) \leq C \Gamma \abs{\beta - \beta^\prime} ~\mbox{.}
\end{equation}
where $\Gamma = 4 \pi (\sum_{n \in \mathbb{Z}} n^2 \mathrm{e}^{-2 \gamma n})^\frac{1}{2}$. 
\end{lemma}
\begin{proof}
From
\begin{multline}
\left\Vert (\hat{H}_{\beta^\prime,\xi} - E) \psi \right\Vert^2  \leq  \left\Vert (\hat{H}_{\beta^\prime,\xi} - \hat{H}_{\beta,\xi}) \psi \right\Vert^2 \\ 
 \leq 4 \sum_{n \in \mathbb{Z}} \left \vert \cos(2 \pi (\beta n + \xi)) - \cos(2 \pi (\beta^\prime n + \xi)) \right\vert^2 \abs{\psi_n}^2 
 \leq C^2 \Gamma^2  \abs{\beta - \beta^\prime}^2 ~\mbox{,}
\end{multline}
one concludes that
\begin{equation}
\mathrm{dist}(E, \hat{\sigma}(\beta^\prime,\xi)) \leq C \Gamma \abs{\beta - \beta^\prime} ~\mbox{.}
\end{equation}
Hence, (\ref{eq_S+dual1}) follows from (\ref{eq_S+dual}).
\end{proof}

We are now ready to prove Proposition \ref{prop_case2}:
\begin{proof}[Proof of Proposition \ref{prop_case2}]
Let $\alpha$ Diophantine. Given $M \in \mathbb{N}$, we define
\begin{equation}
\mathcal{R}_{\alpha,M}^{(l)}:=\{E \in \mathcal{R}_\alpha^{(l)}: ~\mbox{E is $(M, 2 \pi \frac{1}{l})$-localized for $\hat{H}_{\alpha,\theta}$, some $\theta \in \mathbb{T}$}\} ~\mbox{.}
\end{equation}

By Lemma \ref{lemma_dual}, 
\begin{equation}
\mathcal{R}_\alpha^{(l)} = \bigcup_{M} \mathcal{R}_{\alpha, M}^{(l)} ~\mbox{.}
\end{equation}

Let $M\in \mathbb{N}$ be fixed. It suffices to show that,
\begin{equation} \label{eq_wtscase2}
\left\vert \{ E \in \mathcal{R}_{\alpha, M}^{(l)}: \abs{a_{q_n,0}} > 2 + \mathrm{e}^{- c_l q_n} ~\mbox{i.o.} \} \right\vert = 0 ~\mbox{.}
\end{equation}

With a Borel-Cantelli argument in mind, for $n\in \mathbb{N}$, we aim to estimate the measure of the set \footnote{Measurability of the sets $\mathcal{R}_{\alpha, M}^{(l)}$ reduces by Lemma \ref{lemma_dual} to measurability (in $E$) of the matrix-valued function $N(x) = N^E(x)$ which, as remarked earlier, follows from the construction in \cite{H} (see also the footnote following (\ref{eq_defral})).}
\begin{equation}
\Omega_n := \{E \in \mathcal{R}_{\alpha, M}^{(l)}: \abs{a_{q_n,0}(E)} > 2 + \mathrm{e}^{- c_l q_n} \} ~\mbox{.}
\end{equation}

To this end, notice that by (\ref{eq_chambers}), 
$\Omega_n \cap S_+(\frac{p_n}{q_n}) = \emptyset$, and $S_+(p_n/q_n)$ is a union of $q_n $ closed intervals.
Since, for any $E \in \Omega_n$, Lemma \ref{lem_conts+} implies that 
\begin{equation}
\mathrm{dist}(E, S_+(p_n/q_n) ) \leq  M \Gamma \left\vert \alpha - \frac{p_n}{q_n}\right\vert ~\mbox{,}
\end{equation}

we have,
\begin{equation} \label{eq_dualitybc}
\left\vert \Omega_n \right\vert \leq 2 M \Gamma \left\vert \alpha - \frac{p_n}{q_n} \right\vert q_n \leq 2 M \Gamma \dfrac{1}{q_{n+1}} ~\mbox{.}
\end{equation}
Using (\ref{eq_cf2}), Proposition \ref{prop_case2} follows from (\ref{eq_dualitybc}) by Borel-Cantelli.

If $\alpha$ is non-Diophantine, the arguments developed in this section do not apply. However, (\ref{eq_liouv}) implies that 1/2-H\"older continuity of $S_+$ in the Hausdorff metric \cite{A} is enough to conclude (\ref{eq_wtscase2}). By (\ref{eq_cf1}) and (\ref{eq_liouv}), estimating $|\Omega_n|$ using 1/2-H\"older continuity of $S_+$ implies,
\begin{equation}
\left\vert \Omega_n \right\vert \leq 2 C \left\vert \alpha - \frac{p_n}{q_n} \right\vert^{1/2} q_n \leq 2 C \sqrt{\frac{q_n}{q_{n+1}}} < 2 C \sqrt{\frac{1}{q_{n}^{r-1}}} ~\mbox{,}
\end{equation}
which, since $r>1$, is summable by (\ref{eq_cf2}). Here, $C$ is a constant only depending on $v$ (see \cite{A}).
\end{proof}

\section{Polynomials with distinct real roots and the level sets of the discriminants} \label{sec_polyn}

By Fact \ref{fact_floquet} the discriminant, and hence also its average $a_{q_n,0}(E)$, is an algebraic polynomial in $E$ of degree $q_n$ with $q_n$ real distinct roots. In view of Case 3, the purpose of this section is to establish measure estimates for level sets of such polynomials.

Given a Borel-measurable function $f$, for $a < b$ we consider the the measure of the $(a,b)$-level set of $f$,
\begin{equation}
\left \vert f^{-1} (a,b) \right\vert = \left \vert \{ x \in \mathbb{R}: a < f(x) < b \} \right \vert ~\mbox{.}
\end{equation}

For $n \in \mathbb{N}_0$, let $\mathcal{P}_n(\mathbb{R})$ denote the polynomials over $\mathbb{R}$ of {\em{exact}} degree $n$. Given $p \in \mathcal{P}_n(\mathbb{R})$, $\mathrm{LC}(p)$ will stand for the leading coefficient of $p$. A well-known theorem due to P\'olya \cite{BB} states that for $n \geq 1$,
\begin{equation} \label{eq_polya}
\left\vert p^{-1}(a,b) \right\vert = \left\vert \left\{ x \in \mathbb{R}: \left\vert p(x) - \frac{a+b}{2} \right\vert \leq \frac{b-a}{2} \right\} \right\vert \leq 2^{2 - \frac{2}{n}} \left(\dfrac{b-a}{\mathrm{LC}(p)}\right)^{\frac{1}{n}}~\mbox{,}
\end{equation}
for any $p \in \mathcal{P}_n(\mathbb{R})$. This can be considered a global version of the fact that locally $p$ cannot behave worse than $x^n$.

Following we want to consider polynomials with a restricted zero set, more specifically the class $\mathcal{P}_{n;n}(\mathbb{R})$ of elements in $\mathcal{P}_n(\mathbb{R})$ with $n$ {\em{distinct}} real zeros. A simple argument shows that for $p \in \mathcal{P}_{n;n}$, $p^\prime$ and $p^{\prime\prime}$ cannot both be zero at any given point. Thus, locally, such polynomial will behave at worst quadratically, which would lead one to expect that the $\frac{1}{n}$-power dependence on $(b-a)$ on the right hand side of (\ref{eq_polya}) is changed to $\frac{1}{2}$ for elements in $\mathcal{P}_{n;n}(\mathbb{R})$. The interesting fact is that one can obtain a global estimate, with no additional information on $p^\prime$ and $p^{\prime\prime}$.

This intuition is made precise in the Theorems \ref{thm_sublevel_rII} and \ref{thm_poly}.  Both these theorems will take advantage of the ``general structure'' of elements of $\mathcal{P}_{n;n}$, more specifically that any $p \in \mathcal{P}_{n;n}$ has precisely $n-1$ distinct local extrema which, alternatingly, are maxima and minima, respectively. Let $y_1 < \dots < y_{n-1}$ be the local extrema of $p$. Define
\begin{equation} \label{eq_zeta}
\zeta(p) := \min_{1 \leq j \leq n-1} \abs{p(y_{j})} ~\mbox{.}
\end{equation}
For the discriminant, Fact \ref{fact_floquet} (iii) for instance implies
\begin{equation}
\zeta(t_{p/q}(\theta,.)) \geq 2 ~\mbox{,} ~\forall \theta \in \mathbb{T} ~\mbox{.}
\end{equation}

Further, denote by $\mathfrak{Z}(p)$ the zero set of $p$.

Note that any $a$ with $-\zeta(p) \leq a \leq \zeta(p)$ is attained at exactly $n$ distinct points, whereas for $\zeta(p) < \abs{a}$ the multiplicity of $p-a$ is strictly less than $n$. 

\begin{theorem} \label{thm_sublevel_rII}
Let $p \in \mathcal{P}_{n;n}(\mathbb{R})$ and $0 \leq  a < b$.  Then,
\begin{eqnarray} \label{eq_sublevel_rII}
\left\vert p^{-1}(a,b) \right\vert & \leq & 2 \mathrm{diam} \left(\mathfrak{Z}(p-a)\right) \max\left \{ \dfrac{b-a}{\zeta(p)+a}, \left(\dfrac{b-a}{\zeta(p)+a}\right)^\frac{1}{2} \right\} 
\end{eqnarray}
\end{theorem}
\begin{remark}
\begin{itemize}
\item[(i)] An estimate analogous to (\ref{eq_sublevel_rII}) holds for the case $a < b \leq 0$, where $\mathfrak{Z}(p-a)$ and $(\zeta(p) + a)$ are replaced, respectively, by $\mathfrak{Z}(p-b)$ and $(\zeta(p)+b)$.
\item[(ii)] Note that if $a \leq \zeta(p)$, $\mathrm{diam} \left(\mathfrak{Z}(p-a)\right) \leq \mathrm{diam} \left(\mathfrak{Z}(p-\zeta(p))\right)$.
\end{itemize}
\end{remark}

In view of Case 3, introduced in Sec. \ref{sec_outline}, application of Theorem \ref{thm_sublevel_rII} to the discriminant immediately yields
\begin{coro} \label{coro_sublevel_rII}
Let $p/q \in \mathbb{Q}$, $(p,q)=1$, $0 \leq a \leq 2 $, and $a < b$. For all $\theta \in \mathbb{T}$ we have
\begin{equation}
\left\vert t_{p/q}(\theta, .)^{-1}(a,b) \right\vert \leq 4 (2 + \norm{v}_\mathbb{T}) \max\left \{ \dfrac{b-a}{2}, \left(\dfrac{b-a}{2}\right)^\frac{1}{2} \right\} ~\mbox{.}
\end{equation}
\end{coro}

\begin{proof}[Proof of Theorem \ref{thm_sublevel_rII}]
For $p \in \mathcal{P}_{n;n}(\mathbb{R})$, fix $0 \leq a < b$. Let $x_j, ~1\leq j \leq r$, denote the (distinct) roots of $p(x)-a$, $1\leq r \leq n$. Set $f(x) := p(x) - a$, then $x_j, ~1\leq j \leq r$, constitute the roots of $f(x)$. Let $y$ such that $0 \leq f(y) \leq b-a$. Consider first the case that $x_{j} < y < x_{j+1}$ for some (unique) $1 \leq j \leq r$. Let $z_1, z_2 \in p^{-1}(\{-\zeta(p)\})$, $z_1 < z_2$, be closest to $x_j, x_{j+1}$ (so $x_{j-1} < z_1 < z_2 < x_{j+1}$). For $z \in \{z_1, z_2\}$ write,
\begin{equation} \label{eq_sublevel_rII_1}
f(z) = \dfrac{f(z)}{f(y)} f(y) = f(y) ~\dfrac{Q(z)}{Q(y)} ~\dfrac{(z-x_j)(z-x_{j+1})}{(y-x_j)(y-x_{j+1})} ~\mbox{,}
\end{equation}
where 
\begin{equation}
Q(x):= \prod_{k \neq j, j+1} (x-x_k) ~\mbox{.}
\end{equation}

Since $\abs{Q(x)}$ is non-zero with a unique critical point on $(x_{j-1}, x_{j+2})$, we have 
\begin{equation}
\abs{Q(y)} \geq \min_{j=1,2} \abs{Q(z_j)} =: \abs{Q(z_\alpha)} ~\mbox{.}
\end{equation} 
Thus,
\begin{equation}
\zeta(p) + a = \abs{f(z_\alpha)} \leq \abs{f(y)} ~\dfrac{\vert (z_\alpha-x_j)(z_\alpha-x_{j+1})\vert}{\vert (y-x_j) (y-x_{j+1})\vert } ~\mbox{.}
\end{equation}
Using $f(y) \leq b-a$, we obtain control of the distance of $y$ to at least one of $x_j, x_{j+1}$, 
\begin{equation}
\min_{k=j , j+1} \abs{y-x_k} \leq \vert(z_\alpha-x_j)(z_\alpha-x_{j+1})\vert^\frac{1}{2} \left(\dfrac{b-a}{  \zeta(p) + a    }\right)^\frac{1}{2} ~\mbox{,}
\end{equation}
whence
\begin{eqnarray}
\left\vert \left\{y \in [x_j, x_{j+1}] : \abs{f(y)} \leq (b-a) \right\} \right\vert \leq 2  \vert(z_\alpha-x_j)(z_\alpha-x_{j+1})\vert^\frac{1}{2} \left(\dfrac{b-a}{  \zeta(p)+ a    }\right)^\frac{1}{2} \nonumber \\
\leq \left(\vert z_\alpha -x_j \vert + \vert z_\alpha -x_{j+1}\vert\right) \left(\dfrac{b-a}{  \zeta(p) + a    }\right)^\frac{1}{2} \label{eq_measestbands1}
\end{eqnarray}

Finally, consider the case $y < x_1$ or $x_r < y$. Denoting by $z_1 \in p^{-1}(\{-\zeta(p)\})$ the point closest to $x_j$, where $j \in \{1,r\}$, we write in analogy to (\ref{eq_sublevel_rII_1}) ,
\begin{eqnarray} \label{eq_sublevel_rII_2}
\zeta(p)+a=\abs{f(z)} = \dfrac{\abs{f(z)}}{\abs{f(y)}} \abs{f(y)} = \abs{f(y)} ~\dfrac{\abs{\widetilde{Q}(z)}}{\abs{\widetilde{Q}(y)}} ~\dfrac{\abs{z-x_j}}{\abs{y-x_j}} ~\mbox{,} \\
~\widetilde{Q}(x):=\prod_{k \neq j} (x-x_k) ~\mbox{.}
\end{eqnarray}

$\abs{\widetilde{Q}(x)}$ is increasing (decreasing) for $x\geq x_{r-1}$ ($x\leq x_2$) thus
\begin{equation} 
\abs{y-x_j} \leq \abs{f(y)} \dfrac{\abs{z_1 - x_j}}{ \zeta(p) + a   } ~\mbox{.} \label{eq_measestbands2}
\end{equation}

Taking the sum of all the terms, we obtain the claim of the theorem.
\end{proof}

\subsection{Measure of level sets is extremized at Chebyshev polynomials} \label{subsec_cheby}

Even though above proof was carried out separately within the bands, the resulting measure estimate for the contribution of the individual bands to $\left\vert p^{-1}(a,b) \right\vert$ still depends on the specifics of the polynomial ($\abs{x_j - x_{j+1}}$, see (\ref{eq_measestbands1}) and (\ref{eq_measestbands2}), respectively). Given that for $a > \zeta(p)$, the multiplicity of the roots of $p(x) - a$ depends on both $a$ and $p$, this is plausible. 

For level sets, however,  where $-\zeta(p) \leq a < b \leq \zeta(p)$, one can refine above result to obtain universal bounds for the contributions of each of the bands to  $\left\vert p^{-1}(a,b) \right\vert$. Even though not needed for the purpose of the present article, we consider the result to be of general interest whence state it in Theorem \ref{thm_poly} below.

Observe that for $p \in \mathcal{P}_{n;n}(\mathbb{R})$ and  $-\zeta(p) \leq a < b \leq \zeta(p)$, the set $\{ x \in \mathbb{R} : a \leq p(x) \leq b\}$ splits into $n$ closed intervals, referred to as {\em{bands}}, which may touch only at boundary points. Let $B_j[p](a,b)$ denote the $j$th band, where $1\leq j \leq n$ increases from right to left. 

Like (\ref{eq_polya}), the proof of Theorem \ref{thm_poly} is based on approximation theory. Below mentioned  proof develops further an argument of Shamis and Sodin \cite{D}. We mention that similar ideas have been explored earlier by Peherstorfer and Schiefermayr \cite{E}. 

To this end, let $T_n(x)$ be the $n$th Chebyshev polynomials of the first kind, i.e.
\begin{equation} \label{eq_chebyshev}
T_n(x) = \cos\left( n \arccos(x) \right)~\mbox{,}~ x \in [-1,1] ~\mbox{.}
\end{equation}
The polynomials $T_n$ are the archetype for the class $\mathcal{P}_{n;n}(\mathbb{R})$. In fact, we will argue that for any $p \in \mathcal{P}_{n;n}(\mathbb{R})$, the contribution of each band to  $\left\vert p^{-1}(a,b) \right\vert$ is dominated by those of certain rescaled Chebyshev polynomials.  While many extremal properties of Chebyshev polynomials are well known, we did not find this one in the literature. 

The following Lemma quantifies $\left\vert B_j[T_n](a,b) \right\vert$. Recall that $\mathrm{LC}(T_n) = 2^{n-1}$. For $1 \leq j \leq n$ and $x \in [-1,1]$, set
\begin{equation}
g_j^{(n)}(x):=\begin{cases}
\frac{j \pi - \arccos(x)}{n} & \mbox{, $j$ even,} \\
\frac{(j-1) \pi + \arccos(x)}{n} & \mbox{, $j$ odd.}
\end{cases}
\end{equation}
\begin{lemma}
For $-1 \leq a < b \leq 1$,
\begin{equation} \label{eq_chebyshev_meas}
\left\vert B_j[T_n](a,b) \right\vert = \left\vert \cos(g_j^{(n)}(b)) - \cos(g_j^{(n)}(a)) \right\vert \mbox{.}
\end{equation}
\end{lemma}

\begin{proof}
A straightforward computation based on (\ref{eq_chebyshev}). 
\end{proof}

\begin{remark}
The explicit dependence of (\ref{eq_chebyshev_meas}) on $(b-a)$ is easiest seen from
\begin{equation} \label{eq_chebyshev_measest1}
\left\vert \cos(g_j^{(n)}(b)) - \cos(g_j^{(n)}(a)) \right\vert = \left\vert  \dfrac{1}{n} \int_a^b \dfrac{\sin(g_j^{(n)}(x))}{\sqrt{1-x^2}} \ud x \right\vert
\end{equation}

Based on (\ref{eq_chebyshev_meas}), one may estimate to conclude,
\begin{equation} \label{eq_chebmeasest2}
\left\vert B_j[T_n](a,b) \right\vert \leq \left\vert \dfrac{1}{n} \int_{a}^{b} \dfrac{g_j^{(n)}(x)}{\sqrt{1-x^2}} \ud x \right\vert \leq \dfrac{2}{n} \sqrt{b-a} ~\mbox{,}
\end{equation}
for $0 \leq a < b$. The factor of 2 in the final estimate in (\ref{eq_chebmeasest2}) may be easily improved. For $b<1$ or $b=1$ and $j=0,n$ (``extremal bands''), the first inequality in (\ref{eq_chebmeasest2}) shows that the estimate becomes linear in $b-a$.
\end{remark}

Remarkably, the measure of level sets for all $p \in \mathcal{P}_{n;n}$, within each band, is extremized by a rescaled $T_n.$ For $p\in \mathcal{P}_{n;n}$, we define the scaling factor $c_p:=\left(\frac{2 \mathrm{LC}(p)} {\zeta(p)}\right)^{1/n}.$ Then
\begin{theorem} \label{thm_poly}
Let $p \in \mathcal{P}_{n;n}$. Then,
given $-\zeta \leq a < b \leq \zeta$, the $j$th band satisfies
\begin{equation} \label{eq_poly_meas}
\left\vert B_j [p](a,b) \right\vert \leq \left\vert B_j [T_n\left(c_p \frac{x}{2}\right)] \left(\frac{a}{\zeta(p)}  , \frac{b}{\zeta(p)} \right) \right\vert ~\mbox{.}
\end{equation}
\end{theorem}
In particular, for discriminants the estimate (\ref{eq_chebmeasest2}) implies:
\begin{coro}
Let $p/q \in \mathbb{Q}$, $(p,q)=1$, $0 \leq a \leq 2 $, and $0 \leq a < b \leq 2$. For all $\theta \in \mathbb{T}$ we have
\begin{equation}
\left\vert B_j [t_{p/q}(\theta, .)](a,b) \right\vert \leq \dfrac{2^{3/2}}{n} \sqrt{b-a} ~\mbox{,}
\end{equation}
for all $1 \leq j \leq q$.\footnote{With, as before, a linear estimate for $b<2$ or $j=0,n.$}
\end{coro}
\begin{remark}
Theorem \ref{thm_poly} is a generalization of a result of Shamis and Sodin \cite{D}, which is the same type of statement for the full spectral bands of the discriminants of Jacobi operators.
\end{remark}

\begin{proof}
First observe that without loss of generality $\mathrm{LC}(p) > 0$ (if not, take $-p(x)$). 

\begin{figure} 
\includegraphics[width=\textwidth]{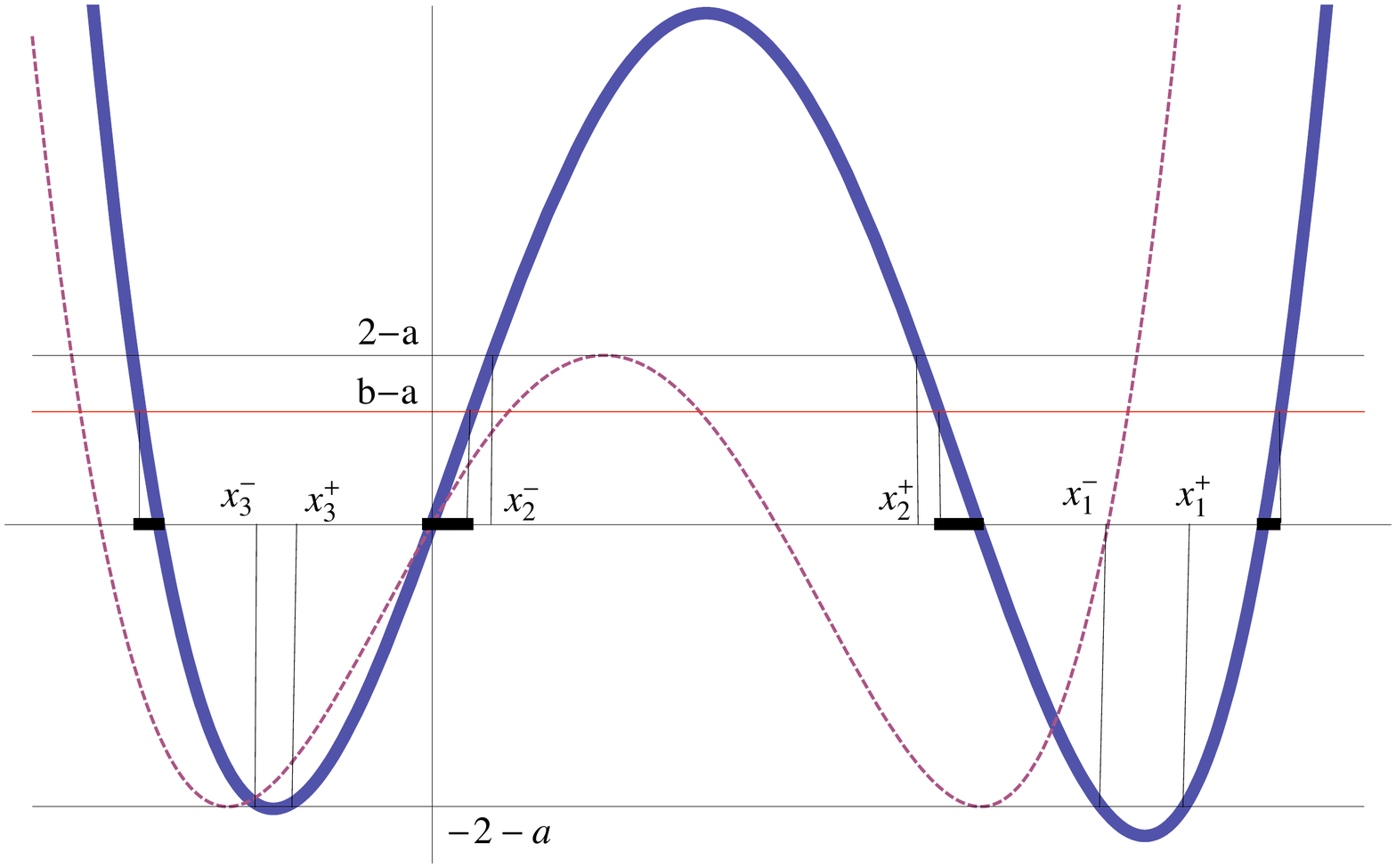} 
\caption{The solid line shows an arbitrary element $\mathrm{LC}(p) x \tau(x;x_1,x_2,x_3) \in \mathcal{F}$ for the case $n=4$, $j=3$ and $\zeta = 2$. The dashed line represents the extremal element $f_0$ given in (\ref{eq_extremalpoly}), for which $x_j^+ = x_j^-$ for all $1 \leq j \leq 3$.} \label{figure_1}
\end{figure}

Fix a band $B_j$, $1 \leq j \leq n$. Let $y_j, \tilde{y_j}$ be the boundary points of $B_j$ such that  $p(y_j) = a$ and $p(\tilde{y_j}) = b$, respectively. Without loss of generality we may assume $y_j = 0$. In particular, we then have $\tilde{y_j} > 0$ if $j$ is odd and $\tilde{y_j} < 0$ if $j$ is even. See Fig. \ref{figure_1} for illustration.

Since vertical shifts of a function do not affect the measure of its level sets, we will apply the line of argument of \cite{D} to the shifted polynomial $q(x) = p(x) - a$. Notice that $q(y_j) = 0$. 

Following, we consider a family $\mathcal{F} $ of deformations  $h \in \mathcal{P}_{n:n}$ of $q$ which all are zero at the reference point $y_j=0$, have $\mathrm{LC}(h)=\mathrm{LC}(p)$, and $\zeta(h+a)\geq \zeta(p).$ Then $\mathcal{F} $ can be represented as 
\begin{equation}
\mathcal{F} =\{ \mathrm{LC}(p) x \tau(x; x_1, \dots, x_{n-1}) , ~ x_{n-1} < x_{n-2} < \dots < x_1\} ~\mbox{,}
\end{equation}
where $\tau(x; x_1, \dots, x_{n-1})$ is the Lagrange interpolating polynomial determined by the conditions
\begin{eqnarray} \label{eq_interpol}
\mathrm{LC}(\tau(.; x_1, \dots x_{n-1})) = 1, \\ \tau(x_k; x_1, \dots x_{n-1}) = (-1)^k \dfrac{\zeta(p) + (-1)^{k+1} a}{\mathrm{LC}(p) x_k} =: (-1)^k \dfrac{\eta_k}{\mathrm{LC}(p) x_k} ~\mbox{.}
\end{eqnarray}
Notice that $\tau(x;x_1,\dots,x_k)$ is defined so that $f \in \mathcal{F}$ satisfies $f(y_j) = 0$ and $f(x_k) = (-1)^k \eta_k$ (see Fig.  \ref{figure_1}), which implicitly defines the points $x_k$. In particular, $q(x) \in \mathcal{F}$. 

It is important to realize that the correspondence between $x_1 > \dots > x_{n-1}$ and members of $\mathcal{F}$ is in general not unique. Given an element $f \in \mathcal{F}$, there may be up to $2^{n-1}$ choices for $x_1 > \dots > x_{n-1}$ representing $f$. This follows from the fact that given $f$, there are in general two possibilities for each $x_k$, denoted $x_k^- \leq x_k^+$ (see Fig.  \ref{figure_1}). 

The crucial observation, however, is that there is a {\em{unique}} member of $\mathcal{F}$ which satisfies $x_k^{-} = x_k^{+}$, for all $1 \leq k \leq n-1$; this is implied by the following Lemma:
\begin{lemma} \label{lem_cheb}
Let $L, \zeta > 0$ be fixed. For $n \in \mathbb{N}$, up to a horizontal shift, there exists a unique $p \in \mathcal{P}_{n;n}$ with $LC(p) = L$, such that at {\em{all}} of its local extrema $\abs{p} = \zeta$.  It is given by
\begin{equation}
p(x) = \zeta T_n\left( \left(\dfrac{L}{\zeta}\right)^{1/n} \dfrac{x}{2^{1-1/n}} \right) ~\mbox{.}
\end{equation}
\end{lemma}
\begin{remark}
Lemma \ref{lem_cheb} follows easily from the standard proof of Chebyshev's alternation theorem. For completeness, we give the simple argument in Appendix  \ref{app_lemcheb}.
\end{remark}

In particular, Lemma \ref{lem_cheb} identifies the distinguished member where $x_k^{-} = x_k^{+}$ for all $k$,  as
\begin{equation} \label{eq_extremalpoly}
f_0(x):= \zeta(p) T_n \left(  c_p\dfrac{x + \delta_j(a)}{2} \right) - a ~\mbox{,}
\end{equation}
a Chebyshev polynomial, shifted and rescaled so that its leading coefficient equals $\mathrm{LC}(p)$, it oscillates between $(\pm \zeta(p) - a)$, and that $y_j=0$. The last condition implicitly defines $\delta_j(a)$ as the $j$th root of the equation 
\begin{equation}
\zeta(p) T_n \left( c_p\dfrac{x}{2} \right) = a ~\mbox{,}
\end{equation}
where as before $j$ is an index increasing from right to left, i.e. $\delta_n(a) < \delta_{n-1}(a) < \dots < \delta_{j}(a) < \dots < \delta_1(a)$.

In \cite{D}, Shamis and Sodin essentially analyze the dependence of elements of $\mathcal{F}$ on the parameters $x_1, \dots, x_{n-1}$. More specifically, it is shown that for any $x^* \in B_j$ one has:
\begin{lemma} \label{lem_shamissodin}
\begin{equation}
\min_{ f \in \mathcal{F}} \abs{f(x^*)} = \abs{f_0(x^*)} ~\mbox{.}
\end{equation}
\end{lemma}
As mentioned above the analysis in \cite{D} was carried out for a special case, however generalizes to our set-up with only a few changes. For the reader's convenience, we give a proof of Lemma \ref{lem_shamissodin} in Appendix \ref{app_shamissodin}.

Thus,
\begin{eqnarray}
\left\vert \left\{ x \in B_j: a < p(x) \leq b \right\} \right\vert & = & \left\vert \left\{ x \in B_j: 0 < q(x) \leq b-a \right\} \right\vert \nonumber \\
 & \leq & \left\vert \left\{ x \in B_j: 0 < f_0(x) \leq b-a \right\} \right\vert ~\mbox{,} \label{eq_chebyband} 
\end{eqnarray}
which proves the claim of the theorem.
\end{proof}

\section{Proof of Theorem \ref{thm_main}} \label{sec_case3}

In order to finish the proof of Theorem \ref{thm_main2} (and thus of Theorem \ref{thm_main}), we are left to consider Case 3 introduced in Sec. \ref{sec_outline}:
\begin{prop} \label{prop_case3}
\begin{equation} \label{eq_final}
\left \vert \left\{ E \in \mathcal{R}_\alpha^{(l)}: E ~\mbox{satisfying Case 3} \right\} \right \vert = 0 ~\mbox{.}
\end{equation}
\end{prop}
\begin{proof}
Using Proposition \ref{prop_chambers}, further decompose $\mathcal{R}_\alpha^{(l)}$ into the countable union of
\begin{equation}
\mathcal{R}_\alpha^{(l),N}:=\{ E \in \mathcal{R}_\alpha^{(l)}: \sup_{\theta \in \mathbb{T}} \left\vert t_{p_n/q_n}(\theta, E) - a_{q_n,0}(E) \right\vert \leq \mathrm{e}^{- c_l q_n} ~\mbox{,  for} ~n \geq N\} ~\mbox{,}
\end{equation}
for $N \in \mathbb{N}$. \footnote{The sets $\mathcal{R}_\alpha^{(l),N}$ are indeed measurable since they are intersections of $ \mathcal{R}_\alpha^{(l)}$ with measurable sets.} It suffices to show, $|\mathcal{R}_\alpha^{(l),N} \cap \mbox{Case 3} |=0$, for each $N \in \mathbb{N}$.

Fix $N \in \mathbb{N}$. For all $n>N$, Proposition \ref{prop_chambers} implies that for all $\theta \in \mathbb{T}$,
\begin{equation} \label{eq_case3red}
\{ E \in \mathcal{R}_\alpha^{(l),N} : \abs{a_{q_n}(E) \pm 2} \leq  \mathrm{e}^{- c_l q_n} \} \subseteq \left \{E: \left\vert t_{p_n/q_n}(\theta,E) \pm 2 \right\vert \leq 2\mathrm{e}^{- c_l q_n} \right \} ~\mbox{.}
\end{equation}

Hence, with a Borel-Cantelli argument in mind, it suffices to estimate the measure of the right hand side of (\ref{eq_case3red}) for {\em{some}} fixed $\theta_0 \in \mathbb{T}$. This estimate is taken care of by Corollary \ref{coro_sublevel_rII}, whence
\begin{equation}
\left\vert \left \{E: \left\vert t_{p_n/q_n}(\theta_0,E) \pm 2 \right\vert \leq  2\mathrm{e}^{- c_l q_n} \right \} \right\vert \leq C (2 + \norm{v}_\mathbb{T}) \mathrm{e}^{-c_l q_n/2} ~\mbox{,}
\end{equation}
which is summable in $n$ for each fixed $l \in \mathbb{N}$.
\end{proof}

We mention that the proof of Proposition \ref{prop_case3} together with Proposition \ref{prop_case2}, implies the following reformulation of Theorem \ref{thm_main2}, which we believe to be of independent interest:
\begin{theorem} \label{thm_mainindividphase}
Given $\alpha$ irrational. For all $l \in \mathbb{N}$ and a.e. $E \in \mathcal{R}_\alpha^{(l)}$ we have
\begin{equation}
\sup_{\theta \in \mathbb{T}} \abs{t_{p_n/q_n}(\theta,E)} \leq 2 - 2 \mathrm{e}^{-c_l q_n} ~\mbox{, eventually.}
\end{equation}
\end{theorem}

In order to prove Theorem \ref{thm_main} (ii), first note that continuity of $S_+(\beta)$ in Hausdorff metric \cite{AVS} implies
\begin{equation} \label{lem_thm_main_3}
\limsup_{\frac{p}{q} \to \alpha} S_+\left(\frac{p}{q}\right) \subseteq \Sigma(\alpha)  ~\mbox{,}
\end{equation}
for any irrational $\alpha \in \mathbb{T}$ (inclusion holds set-wise). 

For the remainder of the proof of Theorem \ref{thm_main} (ii), we have to distinguish between Diophantine and non-Diophantine $\alpha$. Similar to Sec. \ref{sec_dualtiy}, it is the modulus of continuity of $S_+$ in the Hausdorff metric which requires separate treatment of these two cases. 

We start with $\alpha$ Diophantine. Will make use of the following result, established in \cite{JK}, which we formulate in a way useful to the present application:
\begin{theorem}[\cite{JK}; Theorem 3 and Remark 1.] \label{thm_jk}
Let $\alpha \in \mathbb{T}$ be Diophantine. For $\eta >0$, consider the set $\mathscr{E}_\eta:=\{E \in \Sigma(\alpha): L(\alpha, A^E) \geq \eta\}$. There exist $h(\alpha, \eta)>0$, $c(\alpha, \eta) <\infty$ and $\gamma(\alpha) \geq 3$ such that for any $E \in \mathscr{E}_\eta$ and $\beta \in \mathbb{T}$ with $\abs{\alpha - \beta} < h(\alpha, \eta)$,
\begin{equation}
\mathrm{dist}(E, \Sigma_+(\beta)) \leq c(\alpha, \eta) \left\vert (\alpha -\beta) \log^{\gamma(\alpha)} \abs{\alpha - \beta} \right\vert ~\mbox{.}
\end{equation}
\end{theorem}
\begin{remark}
$\gamma$ depends on $\alpha$ only through its Diophantine constants. For $\alpha \in \mathbb{T}$ whose continued fraction expansion forms a {\em{bounded}} sequence, $\gamma (\alpha) = 3$.
\end{remark}

Using (\ref{lem_thm_main_3}), we are left to show that 
\begin{equation}
\Sigma(\alpha) \subseteq \liminf_{p_n/q_n \to \alpha} S_+\left(\frac{p_n}{q_n}\right) ~\mbox{.}
\end{equation}
Employing Remark \ref{rem_conjtorot}(i), we partition $\Sigma(\alpha)$ according to
\begin{equation}
\Sigma(\alpha) = \mathcal{R}_\alpha \bigcup \left\{ E \in \Sigma(\alpha): L(\alpha, A^E) > 0 \right\}  ~\mbox{.}
\end{equation}

Clearly, Theorem \ref{thm_main2} already shows
\begin{equation}
\mathcal{R}_\alpha \subseteq \liminf_{p_n/q_n \to \alpha} S_+\left(\frac{p_n}{q_n}\right) ~\mbox{,}
\end{equation}
whence it remains to prove that
\begin{equation}
\left\{ E \in \Sigma(\alpha): L(\alpha, A^E) > 0 \right\} \subseteq \liminf_{p_n/q_n \to \alpha} S_+\left(\frac{p_n}{q_n}\right) ~\mbox{.}
\end{equation}
In turn, this will follow by showing,
\begin{equation}
\mathscr{E}_{1/k} \subseteq \liminf_{p_n/q_n \to \alpha} S_+\left(\frac{p_n}{q_n}\right) ~\mbox{,}
\end{equation}
for all $k \in \mathbb{N}$.

Let $k \in \mathbb{N}$ be fixed and arbitrary. Note that by continuity of the spectrum in Hausdorff metric, $S_+\left(\frac{p_n}{q_n}\right)$ consists of at most $q_n$ disjoint closed intervals. Thus, employing Theorem \ref{thm_jk}, analogous arguments as in the proof of Proposition \ref{prop_case2} yield
\begin{multline}
\left\vert \left\{ E \in \mathscr{E}_{1/k}: E \not\in S_+\left(\frac{p_n}{q_n}\right) \right\} \right\vert \leq 2 q_n c(\alpha, 1/k) \left\vert \left(\alpha - \frac{p_n}{q_n}\right) \log^{\gamma(\alpha)} \left\vert \alpha - \frac{p_n}{q_n}\right\vert \right\vert \\
\leq 2 c(\alpha, 1/k) \frac{1}{q_{n+1}}\log^{\gamma(\alpha)} \left( q_n q_{n+1} \right) \leq 2^{1+\gamma(\alpha)}  c(\alpha, 1/k) \dfrac{\log^{\gamma(\alpha)}(q_{n+1})}{q_{n+1}} ~\mbox{,}
\end{multline}
which is summable by (\ref{eq_cf2}).

Finally, if $\alpha$ is non-Diophantine, employing the same arguments as in the end of Sec. \ref{sec_dualtiy}, we directly conclude that
\begin{equation}
\left\vert \Sigma(\alpha) \setminus \liminf_{n \to \infty} S_+\left( \dfrac{p_n}{q_n} \right) \right\vert = 0 ~\mbox{,}
\end{equation}
since by 1/2-H\"older continuity of $S_+$ in the Hausdorff metric \cite{A},
\begin{equation}
\left\vert \left\{ E \in \Sigma(\alpha): E \notin S_+\left( \dfrac{p_n}{q_n} \right) \right\} \right\vert \leq 2 C \left\vert \alpha - \frac{p_n}{q_n} \right\vert^{1/2} q_n < 2 C \sqrt{\frac{1}{q_{n}^{r-1}}} ~\mbox{.}
\end{equation}
This completes the proof of Theorem \ref{thm_main} (ii).

\section{Some general facts on Duality} \label{sec_genduality}

The purpose of this final section is to present an approach to $S_+$ and duality through the study of decomposable operators. This leads to a simple proof of Theorem \ref{thm_duality_inv}, which has only been explicit in the literature for the almost Mathieu operator \cite{AVS}. All considerations in this section apply to Schr\"odinger operators with {\em{continuos}} potential $v$.

The following is based on the elegant approach originally introduced for almost Mathieu by Chulaevsky and Delyon \cite{ChuDel}. Later, similar ideas were employed in \cite{FFF,DD, EE}.

Physically, duality may be viewed as a change to ``momentum eigenstates'', thus on a heuristic level giving rise to the correspondence between ``localized states'' and Bloch waves. To make this rigorous we consider the constant fiber direct integral,
\begin{equation}
\mathcal{H}^\prime := \int_\mathbb{T}^\oplus \mathit{l}^2(\mathbb{Z}) \ud \theta ~\mbox{,}
\end{equation}
which, as usual, is defined as the space of $\mathit{l}^2(\mathbb{Z})$-valued, $L^2$-functions over the measure space $(\mathbb{T},\ud \theta)$. For the general theory of fiber direct integrals we refer the reader to e.g. \cite{RS}. 

Let $\beta \in \mathbb{T}$ be fixed. Interpreting $H_{\beta,\theta}$ as fibers of the decomposable operator,
\begin{equation} \label{eq_defopdualDF}
H_\beta^\prime=\int_\mathbb{T}^\oplus H_{\beta,\theta} \ud \theta ~\mbox{,}
\end{equation}
the family $\{H_{\beta,\theta}\}_{\theta \in \mathbb{T}} $ naturally induces an operator on the space $\mathcal{H}^\prime$,
\begin{equation}
(H_\beta^\prime \psi)(\theta, .) := H_{\beta,\theta} \psi(\theta, .) ~\mbox{,}
\end{equation}
with equality viewed in $L^2$. Similarly, with the dual $\{ \hat{H}_{\beta,\theta}  \}_{\theta \in \mathbb{T}}$, defined in (\ref{eq_defdual}), we associate the 
decomposable operator,
\begin{equation}
\hat{H}_\beta^\prime:=\int_\mathbb{T}^\oplus \hat{H}_{\beta,\theta} \ud \theta ~\mbox{.}
\end{equation}

We mention that spectral measures for $H_\beta^\prime$ just amount to spectral averages w.r.t. $\ud \theta$, i.e. given $\psi \in \mathcal{H}_\beta^\prime$, the spectral measure $\ud \mu_\psi$ associated with $\psi$ and $H_\beta^\prime$ is
\begin{equation} \label{eq_spectralaverag}
\ud \mu_\psi = \int_\mathbb{T} \ud \mu_{\psi(\theta)} \ud \theta ~\mbox{.}
\end{equation}
Here, $\ud \mu_{\psi(\theta)}$ is the spectral measure for $\psi(\theta,.)$ and $H_{\beta,\theta}$. Similar holds for the dual $\hat{H}_\beta^\prime$. Within the present framework, an important example of (\ref{eq_spectralaverag}) is the density of states, in which case $\psi(\theta, n) = \delta_{1,n}$.

The correspondence between dual operators is mediated by the unitary, $\mathcal{U}: \mathcal{H}^\prime \to \mathcal{H}^\prime$,
\begin{equation} \label{eq_unit}
(\mathcal{U} \psi)(\eta, m):=\sum_{n \in \mathbb{Z}} \int_\mathbb{T} \ud \theta \mathrm{e}^{2 \pi i m \theta} \mathrm{e}^{2 \pi i n (m\alpha+\eta)} \psi(\theta,n) ~\mbox{.}
\end{equation}
Note that $\mathcal{U}^{-1}$ is obtained from (\ref{eq_unit}) by simply reversing the signs in the exponentials. The unitary (\ref{eq_unit}) had first been introduced in context of the almost Mathieu operator \cite{ChuDel}.
\begin{remark}
We mention that combining (\ref{eq_spectralaverag}) and (\ref{eq_unit}) we immediately conclude invariance of the density of states under duality. In \cite{DD} this had already been established using different means. Another proof of invariance of the density of states using (\ref{eq_unit}), written for almost Mathieu but immediately generalizable, is given in \cite{FFF}.
\end{remark}

Duality is expressed as a unitary equivalence of the operators $H_\beta^\prime$ and $\hat{H}_\beta^\prime$,
\begin{equation} \label{eq_unitequiv}
\mathcal{U}^{-1} H_\beta^\prime \mathcal{U} = \hat{H}_\beta^\prime ~\mbox{.}
\end{equation}
We mention, the computation leading to (\ref{eq_unitequiv}), can be simplified using density of trigonometric polynomials in $\mathcal{C}(\mathbb{T})$, in which case
verification of the following identities suffices:
\begin{eqnarray}
\left(\mathcal{U}^{-1} T_{e_k} \mathcal{U} \psi\right) (\theta,n) & = & \psi(\theta,n-k) ~\mbox{,} \\
\left(\mathcal{U}^{-1} T \mathcal{U} \psi \right) (\theta,n) & = & \mathrm{e}^{-2\pi i (\alpha n + \theta)} \psi(\theta,n) ~\mbox{,}
\end{eqnarray}
where for $k \in \mathbb{Z}$, we define $T_{e_k} , T: \mathcal{H}^\prime \to \mathcal{H}^\prime$,
\begin{eqnarray}
\left(T_{e_k} \psi\right)(\theta,n) & := & \mathrm{e}^{2 \pi i k (\alpha n + \theta)} \psi(\theta,n) ~\mbox{,} \\
\left( T \psi\right)(\theta,n) & := &\psi(\theta, n+1) ~\mbox{.}
\end{eqnarray}
Again, all equations here are interpreted in $L^2$.

Denoting the spectra of $H_\beta^\prime$ and $\hat{H}_\beta^\prime$ by $\sigma^\prime(\beta)$ and $\hat{\sigma}^\prime(\beta)$, respectively, (\ref{eq_unitequiv}) implies
\begin{equation} \label{eq_spectradual}
\sigma^\prime(\beta) = \hat{\sigma}^\prime(\beta) ~\mbox{.}
\end{equation}

The following proposition interprets the sets $S_+(\beta)$ and $\hat{S}_+(\beta)$ as the spectra of the decomposable operators $H_\beta^\prime$ and $\hat{H}_\beta^\prime$. In particular, this shows why these sets are the natural quantities to reflect the spectral properties of the {\em{family}} $\{H_{\beta,\theta}\}_{\theta \in \mathbb{T}} $ and $\{ \hat{H}_{\beta,\theta}  \}_{\theta \in \mathbb{T}}$, respectively. 

\begin{prop} \label{prop_ints+}
Assume $v(\theta)$ is continuous and let $\beta \in \mathbb{T}$. Then,
\begin{eqnarray}
\sigma^\prime(\beta) & = & S_+(\beta) ~\mbox{,} \\ \label{eq_ints+wts}
\hat{\sigma}^\prime(\beta) & = & \hat{S}_+(\beta) ~\mbox{.}
\end{eqnarray}
\end{prop}
\begin{proof}
Since the argument for the dual is analogous, we shall focus on establishing (\ref{eq_ints+wts}). First, recall from the general theory of decomposable operators (see e.g. \cite{RS}, Theorem XIII.85)
that $E \in \sigma^\prime(\beta)$ if and only if $\forall \epsilon > 0$,
\begin{equation}
\left\vert \left\{ \theta \in \mathbb{T}: (E - \epsilon, E + \epsilon) \cap \sigma(\beta,\theta) \neq \emptyset \right\} \right\vert > 0 ~\mbox{.}
\end{equation}
We shall make use of the following standard fact
\begin{fact} \label{fact_contspectr}
Let $\mathcal{H}$ be separable Hilbert space, and denote by $\mathscr{S} \subseteq \mathscr{L}(\mathcal{H})$ the Banach-subspace of bounded self adjoint operators on $\mathcal{H}$. Then, 
\begin{equation}
\rho_H( \sigma(A) , \sigma(B) ) \to 0 ~\mbox{, as} ~A \to B ~\mbox{ in} ~\mathscr{L}(\mathcal{H}) ~\mbox{.}  
\end{equation}
Here, $\rho_H( . , .)$ is the Hausdorff metric.
\end{fact}

Let $E \in S_+(\beta)$, then $E \in \sigma(\beta, \theta_0)$, some $\theta_0 \in \mathbb{T}$. By continuity of the potential, 
\begin{equation}
\Vert H_{\beta,\theta} - H_{\beta, \theta_0} \Vert \to 0 ~\mbox{, as} ~\theta \to \theta_0 ~\mbox{,}
\end{equation}
whence Fact \ref{fact_contspectr} implies that given $\epsilon>0$ there exists $\delta>0$ such that
\begin{equation}
\sigma(\beta, \theta) \cap (E-\epsilon,E+\epsilon) \neq \emptyset ~\mbox{,}
\end{equation}
for all $\abs{\theta - \theta_0} < \delta$. In particular, $E \in \sigma^\prime(\beta)$.

Conversely, suppose $E \in \sigma^\prime(\beta)$. Then, by compactness, for some convergent sequence $\theta_n \to \theta_\infty$ and some $\theta_\infty \in \mathbb{T}$,
\begin{equation}
\mathrm{dist}(E, \sigma(\beta, \theta_n)) \to 0 ~\mbox{, as } n \to \infty ~\mbox{.}
\end{equation}
We claim $E \in \sigma(\beta, \theta_\infty) \subseteq S_+(\beta)$. 

Indeed, using Fact \ref{fact_contspectr} and 
\begin{equation}
\mathrm{dist}(E, \sigma(\beta,\theta_\infty)) \leq \mathrm{dist}(E, \sigma(\beta, \theta_n)) + \rho_H(\sigma(\beta, \theta_n), \sigma(\beta, \theta_\infty)) ~\mbox{,}
\end{equation}
yields $\mathrm{dist}(E,\sigma(\beta, \theta_\infty)) = 0$, as claimed.
\end{proof}

As an immediate corollary we obtain Theorem \ref{thm_duality_inv}. We mention that for irrational $\beta$, this could have also been concluded from invariance of the density of states, which, as mentioned earlier had already been known for general operators of the form (\ref{eq_defn_op}) \cite{DD}. In the present framework it simply follows from (\ref{eq_spectralaverag}). The point here is that we obtain Theorem \ref{thm_duality_inv}, by treating rational and irrational $\beta$ on the same footing. For the almost Mathieu operator, Theorem  \ref{thm_duality_inv}  had been obtained in \cite{AVS}, where rational and irrational $\beta$ were considered separately.

\appendix

\section{Avila's quantization of the acceleration for analytic SL(2,$\mathbb{C}$)-cocycles} \label{app_2}

In this section, we provide a proof of Lemma \ref{lemma_avila}, which as mentioned earlier, is a more detailed version of Avila's theorem on quantization of the acceleration \cite{B}. Since the result is general to analytic $SL(2,\mathbb{C})$-cocycles, following we replace $A^E$ by an arbitrary analytic matrix valued function $D:\mathbb{T} \to SL(2,\mathbb{C})$, extending holomorphically to a neighborhood of $\abs{\im(z)} \leq \delta$, for some fixed $\delta>0$. We set $D_\epsilon(x) := D(x + i\epsilon)$, for $\abs{\epsilon} \leq \delta$.

Given $\beta \in \mathbb{T}$, the Lyapunov exponent of the $SL(2,\mathbb{C})$-cocycle $(\beta,D)$ is defined in analogy to (\ref{eq_defle}). 

Subharmonicity of $L(\beta, D_\epsilon)$ viewed as a function of $\epsilon \in \mathbb{C}$, is easily seen to imply that $L(\beta, D_\epsilon)$ is convex in $\mathrm{Re} (\epsilon)$. This shows existence of the right derivative in (\ref{eq_quantaccel}). In context of his global theory of one-frequency operators \cite{A}, Avila introduces the {\em{acceleration}}
\begin{equation} \label{def_accel}
\omega(\alpha, D_\epsilon):=\frac{1}{2 \pi} D_+\left(L(\alpha, D_\epsilon)\right) ~\mbox{,}
\end{equation}
for a fixed {\em{irrational}} $\alpha\in \mathbb{T}$. 

Finally, we mention that when applying the general result proven below to the Schr\"odinger cocycle $(\alpha, A^E)$, just recall that
\begin{equation}
\norm{A^E}_\delta \leq C (1+\norm{v}_\delta) ~\mbox{, for} ~E \in \Sigma(\alpha) ~\mbox{,}
\end{equation}
which yields the claimed uniformity of Lemma \ref{lemma_avila} over $\Sigma(\alpha)$.

\begin{proof}
For $n \in \mathbb{N}$ let $r_{n}=\frac{p_{n}}{q_{n}}$ with $(p_{n},q_{n})=1$ be {\em{any}} sequence of rationals approximating $\alpha$ (not necessarily the canonical approximants from the continued fraction expansion of $\alpha$). Set
\begin{equation*}
D_n(x):=D(x+(q_{n}-1) r_{n}) \dots D(x) ~\mbox{.}
\end{equation*}
Then,
\begin{equation} \label{eq_rationallyap1}
L(r_{n}, D) = \dfrac{1}{q_{n}} \int_{\mathbb{T}} \log{\rho\left(D_n\right)} \ud x ~\mbox{.}
\end{equation}
Here, $\rho\left(D_n\right)$ denotes the spectral radius of the matrix $D_n$. To simplify notation we write $\rho_n := \rho(D_n)$ and $t_n:=\mathrm{tr}(D_n)$, $n \in \mathbb{N}$.

We first claim that uniformly over  $\abs{\epsilon} \leq \delta$ we have
\begin{equation} \label{eq_qa1}
L(r_n, D_\epsilon) = \dfrac{1}{q_n} \int_{\abs{t_n(.+i\epsilon)} \geq 2} \log \abs{\rho_n(x+ i \epsilon)} \ud x + o(1) \mbox{,}
\end{equation}
as $n \to \infty$. 

We shall make use of the following simple fact for $SL(2, \mathbb{C})$ matrices:
\begin{claim} 
Let $A \in SL(2,\mathbb{C})$, then
\begin{equation} \label{eq_trrho}
\max\left\{1, \frac{1}{2} \abs{\mathrm{tr}(A)} \right\} \leq \rho(A) \leq \left( 1 + \sqrt{2} \right) \max\left\{1, \frac{1}{2} \abs{\mathrm{tr}(A)} \right\}
\end{equation}
\end{claim}
\begin{remark}
Both inequalities in (\ref{eq_trrho}) are sharp as can be seen from
\begin{equation}
\rho(A) = \left\vert \dfrac{\mathrm{tr}(A) + \sqrt{\mathrm{tr}(A)^2 - 4}}{2} \right\vert ~\mbox{,}
\end{equation}
for an appropriate branch of the root, and  taking $A$ with, correspondingly, $\mathrm{tr}(A) = 2i$ and $\mathrm{tr}(A) = 2$.
\end{remark}
\begin{proof}
The lower bound in (\ref{eq_trrho}) for $\rho(A)$ is obvious. The upper bound follows since
\begin{equation}
\abs{\mathrm{tr}(A)} \geq \rho(A) - \frac{1}{\rho(A)} ~\mbox{,}
\end{equation}
which implies that the spectral radius and the trace satisfy
\begin{equation}
\rho^2(A) - \abs{\mathrm{tr}(A)} \rho(A) - 1 \leq 0 ~\mbox{.}
\end{equation}
Thus,
\begin{equation} \label{eq_trrho1}
\rho(A) \leq \frac{1}{2} \left\{ \abs{\mathrm{tr}(A)} + \sqrt{4 + \abs{\mathrm{tr}(A)}^2}    \right\} ~\mbox{,}
\end{equation}
which upon considering separately the two cases $\abs{ \mathrm{tr}(A)} \geq 2$ and $\abs{ \mathrm{tr}(A)} < 2$ yields the rightmost inequality of (\ref{eq_trrho}).
\end{proof}

Equation (\ref{eq_trrho}) shows that $1 \leq \rho_n \leq 1+ \sqrt{2}$ whenever $\abs{t_n} < 2$; hence we conclude,
\begin{equation}
0 \leq \frac{1}{q_n} \int_{\abs{t_n(. + i \epsilon)} < 2} \log \rho_n \ud x \leq C/q_n \to 0 ~\mbox{,}
\end{equation}
uniformly over $\abs{\epsilon} \leq \delta$ as $n \to \infty$, giving rise to (\ref{eq_qa1}).

Notice that $(p_n, q_n) = 1$ implies that $t_n$ is a $q_n$-periodic, analytic function with extension to a neighborhood of $\abs{\im(z)} \leq \delta$. Due to analyticity, it is desirable to replace $\rho_n$ in the integrand of (\ref{eq_qa1}) by $t_n$. To justify this, we employ (\ref{eq_trrho}) and conclude,
\begin{equation}
0 \leq \frac{1}{q_n} \int_{\abs{t_n(. + i \epsilon)}\geq 2} \log \left\vert \dfrac{\rho_n}{t_n} \right\vert \ud x \leq \frac{1}{q_n} \log(1+\sqrt{2}) \to 0 ~\mbox{,}
\end{equation}
uniformly over $\abs{\epsilon} \leq \delta$ as $n \to \infty$. Correspondingly we obtain the following basic expression for the LE,
\begin{equation}
L(r_n, D_\epsilon) = \frac{1}{q_n} \int_{\abs{t_n(.+i \epsilon)}\geq 2} \log \abs{t_n(x+i\epsilon)} \ud x + o(1) ~\mbox{,}
\end{equation}
uniformly in $\abs{\epsilon} \leq \delta$ as $n \to \infty$.

Writing $t_n(x+i \epsilon) :=\sum_{k \in \mathbb{Z}} a_{n,k} \mathrm{e}^{2 \pi i k q_n (x + i \epsilon)}$, analyticity in a neighborhood of $\abs{\im(z)} \leq \delta$ implies the following decay of Fourier-coefficients,
\begin{equation} \label{eq_trrho2}
\abs{a_{n,k}} \leq 2 \norm{D}_\delta^{q_n} \mathrm{e}^{- 2 \pi \abs{k} q_n \delta} ~\mbox{,} ~k \in \mathbb{Z} ~,~n \in \mathbb{N} ~\mbox{.}
\end{equation}
Choosing $K \in \mathbb{N}$ sufficiently large so that 
\begin{equation} \label{eq_trrho3}
2 \pi \delta K - \log \norm{D}_\delta > 0 ~\mbox{,}
\end{equation}
ensures exponential decay of $a_{n,k}$ in (\ref{eq_trrho2}) {\em{independent}} of $n$ for $k \geq K$, i.e. for any fixed $0<\delta_1<\delta$, $\exists$ $K = K(\norm{D}_\delta,\delta_1)$ and a constant $C=C(\delta_1)$ such that
\begin{equation} \label{eq_trrho4}
\max_{x \in \mathbb{T}} \left\vert \sum_{\abs{k} > K} a_{n,k} \mathrm{e}^{2 \pi i q_n k (x+i \epsilon)} \right\vert \leq C \mathrm{e}^{- 2 \pi q_n \delta_1} ~\mbox{,}
\end{equation}
for all $0 \leq \abs{\epsilon} \leq \delta_1$. 

Let $0< \delta_1 < \delta$ and the corresponding $K = K(\norm{D}_\delta,\delta_1)$ be fixed. Furthermore, for $0 \leq \abs{\epsilon} \leq \delta_1$, define $k_n \in \{-K,\dots,K\}$ so that 
\begin{equation}
\left\vert a_{n,k_n} \mathrm{e}^{-2\pi q_n \epsilon k_n} \right\vert = \max_{\abs{k} \leq K} \left\vert a_{n,k} \mathrm{e}^{-2 \pi q_n \epsilon k} \right\vert =: M_n ~\mbox{.}
\end{equation}
Note that both $k_n$ as well as $M_n$ depend on $\epsilon$.

We emphasize the importance of (\ref{eq_trrho4}) in that it allows a cut-off of $t_n$, i.e. uniformly over 
$0\leq \abs{\epsilon} \leq \delta_1$ and $x \in \mathbb{T}$ we obtain
\begin{equation} \label{eq_trrho5}
t_n(x+i\epsilon) = \sum_{\abs{k} \leq K} a_{n,k} \mathrm{e}^{2 \pi i k q_n (x + i \epsilon)} + \mathcal{O}(\mathrm{e}^{-2 \pi q_n \delta_1}) ~\mbox{,}
\end{equation}
as $n \to \infty$. Applied to the cocycle $(\alpha, A^E)$, Eq. (\ref{eq_trrho5}) proves (\ref{eq_chambers3}) of Lemma \ref{lemma_avila}.

\begin{lemma}
Uniformly over $0\leq \abs{\epsilon} \leq \delta_1$ we have,
\begin{equation} \label{eq_trrho7_1}
\frac{1}{q_n} \int_{\abs{t_n(.+i\epsilon)}\geq 2} \log \left\vert \sum_{\abs{k} \leq K} a_{n,k} \mathrm{e}^{2 \pi i k q_n (x+i \epsilon)} \right\vert \ud x= \frac{1}{q_n} \log M_n + o(1) ~\mbox{,}
\end{equation}
as $n \to \infty$.
\end{lemma}
\begin{proof}
Using (\ref{eq_trrho5}), 
\begin{equation} \label{eq_trrho8}
\left\vert \sum_{\abs{k}\leq K} a_{n,k} \mathrm{e}^{2 \pi i k q_n(x+i\epsilon)} \right\vert \geq 1 ~\mbox{,}
\end{equation}
for sufficiently large $n$ (uniformly over $0\leq\abs{\epsilon}\leq \delta_1$) on the set $\{ \abs{t_n(.+i\epsilon)} \geq 2\}$.

Setting $P_n(x):=\sum_{\abs{k} \leq K} a_{n,k} \mathrm{e}^{2\pi i k q_n (x+i \epsilon)}$, it thus suffices to show  that 
\begin{equation}
\frac{1}{q_n} \int_{\abs{P_n(x+i\epsilon)} \geq 1} \log \left( \dfrac{\abs{P_n(x+i\epsilon)}}{M_n} \right) \ud x \to 0 ~\mbox{,}
\end{equation}
uniformly over $0 \leq \abs{\epsilon} \leq \delta_1$ as $n \to \infty$.

First, we note that 
\begin{equation}
\left\vert \dfrac{P_n(x+i\epsilon)}{M_n} \right\vert= \left\vert 1 + \sum_{\stackrel{\abs{k} \leq K}{k \neq k_n}}\dfrac{ a_{n,k} \mathrm{e}^{- 2 \pi \epsilon k  q_n} }{a_{k_n,n} \mathrm{e}^{-2 \pi  \epsilon k_n q_n}} \mathrm{e}^{2 \pi i x q_n (k - k_n)} \right\vert =: \left \vert Q_n(x + i \epsilon) \right\vert ~\mbox{,}
\end{equation}
where $Q_n(x + i \epsilon) = \sum_{j=0}^{2K} c_{j,n}(\epsilon) \mathrm{e}^{2 \pi i x q_n j}$ with
\begin{equation} \label{eq_trrho7}
\abs{c_{j,n}(\epsilon)} \begin{cases}  \leq 1 & \mbox{, if} ~ j-k_n \neq -K ~\mbox{,}  \\
= 1 & \mbox{, if} ~ j-k_n = -K ~\mbox{.}      \end{cases}
\end{equation}

Let $\epsilon \leq \delta_1$ be arbitrary. For $j \in \mathbb{N}_0$, consider the level sets $\Omega_{j,n}:= \{ x\in \mathbb{T}: 1 \leq \abs{P_n(x+i\epsilon)} ~\mbox{,} ~ \mathrm{e}^{-(j+1)} \leq \abs{Q_n(x+i\epsilon)} \leq \mathrm{e}^{-j} \}$. Then,
\begin{equation} \label{eq_trrho6}
\int_{\abs{P_n(.+i\epsilon)}\geq 1} \log \left\vert \dfrac{P_n(x+i\epsilon)}{M_n} \right\vert = \sum_{j=0}^{\infty} \int_{\Omega_{j,n}} \log \abs{Q_n(x+i\epsilon)} \ud x + \int_{\abs{Q_n(x+i\epsilon)} \geq 1,\abs{P_n(x+i\epsilon)} \geq 1} \log \abs{Q_n(x+i\epsilon)} \ud x ~\mbox{.}
\end{equation}

The second contribution on the right hand side of (\ref{eq_trrho6}) is easily dealt with,
\begin{equation}
0 \leq \frac{1}{q_n} \int_{\abs{Q_n(x+i\epsilon)} \geq 1,\abs{P_n(x+i\epsilon)} \geq 1} \log \abs{Q_n(x+i\epsilon)} \ud x \leq \frac{1}{q_n} \log(2 K + 1) \to 0 ~\mbox{.}
\end{equation}

We estimate $\vert \Omega_{j,n} \vert$ using the following well-known Remez-type inequality which e.g. can be obtained from Cartan's Lemma \footnote{We thank Sasha Sodin for enlightening discussions on the history of such statements.}. For a review on statements of this type for algebraic and trigonometric polynomials see e.g. \cite{T}. We mention that a related fact was rediscovered in \cite{Q}, see Theorem 8 therein.
\begin{theorem} \label{ref_polyn}
Let $Q(x) : = \sum_{j=1}^{r} c_j \mathrm{e}^{2 \pi i x j}$ be a polynomial of degree $r$ in the variable  $\mathrm{e}^{2 \pi i x}$. There exists a {\em{universal}} constant $C$ such that for a given measurable set $X \subseteq \mathbb{T}$, $\vert X \vert >0$, the following holds:
\begin{equation}
\norm{Q}_\mathbb{T} \leq (C/\vert X \vert)^{r} \sup_{x \in X} \abs{Q(x)} ~\mbox{.} 
\end{equation}
\end{theorem}

Notice that (\ref{eq_trrho7}) implies $\norm{Q_n(.+i\epsilon)}_\mathbb{T} \geq \abs{c_{k_n - K}} = 1$, hence Theorem \ref{ref_polyn} enables to bound the first term on the right hand side of (\ref{eq_trrho6})
\begin{equation}
\left\vert \sum_{j=0}^{\infty} \int_{\Omega_{j,n}} \log \abs{Q_n(x+i\epsilon)} \ud x \right\vert \leq C \sum_{j=0}^{\infty} (j+1) \mathrm{e}^{-j/(2K+1)} ~\mbox{,}
\end{equation}
which completes the proof of the Lemma.
\end{proof}

Recalling (\ref{eq_trrho5}), bounded convergence accounts for the deviation of $t_n$ from its cut-off $P_n$, since by (\ref{eq_trrho8})
\begin{equation}
\frac{1}{q_n} \log \left\vert 1 \pm \frac{\mathcal{O}(\mathrm{e}^{-2 \pi q_n \delta_1})}{P_n(x+i \epsilon)} \right\vert \to 0 ~\mbox{,}
\end{equation}
uniformly on $\{ \abs{t_n(x+i\epsilon)} \geq 2, ~\abs{\epsilon} \leq \delta_1 \}$.

Equation (\ref{eq_trrho7_1}) is thus improved giving,
\begin{equation} \label{eq_trrho7a}
\frac{1}{q_n} \int_{\abs{t_n(.+i\epsilon)}\geq 2} \log \left\vert t_n(x+i\epsilon) \right\vert \ud x= \frac{1}{q_n} \log M_n + o(1) ~\mbox{,}
\end{equation}
uniformly on $0\leq \abs{\epsilon}\leq \delta_1$, as $n \to \infty$, as $n \to \infty$.

Finally we mention that in principle one could imagine the right hand side of (\ref{eq_trrho7a}) to diverge if $M_n$ becomes arbitrarily close to $0$ as $n \to \infty$. That this is not the case is the subject of the following:
\begin{lemma} \label{lem_qac2}
Let $[\epsilon_1, \epsilon_2] \subseteq [-\delta_1, \delta_1]$ a closed interval so that  $\min_{\epsilon \in [\epsilon_1,\epsilon_2]} L(\alpha, D_\epsilon) > 0$. Then, there exists $N \in \mathbb{N}$ such that $\min_{\epsilon \in [\epsilon_1,\epsilon_2]} M_n(\epsilon) \geq \frac{1}{2 K}$, whenever $n \geq N$.
\end{lemma}
\begin{proof}
Continuity of the LE for non-singular cocycles w.r.t. $r_n \to \alpha \notin \mathbb{Q}$ \cite{C}, implies that $L(r_n, D_\epsilon) \to L(\alpha, D_\epsilon)$ uniformly on $0 \leq \abs{\epsilon} \leq \delta_1$ as $n\to\infty$. Hence, since $\min_{\epsilon \in [\epsilon_1,\epsilon_2]} L(\alpha, D_\epsilon) > 0$, for any given $\epsilon \in [\epsilon_1,\epsilon_2]$, 
$\{ \abs{t_n(x+i\epsilon)}\geq 2 \} \neq \emptyset$ for $n \geq \tilde{N}(\epsilon)$. By compactness of $[\epsilon_1, \epsilon_2]$ this however already produces $N \in \mathbb{N}$ such that for any $n \geq N$, $\{ \inf_{\epsilon \in [\epsilon_1, \epsilon_2]}  \abs{t_n(. + i \epsilon)} > 2\} \neq \emptyset$.

In summary there exists $x_0\in \mathbb{T}$ satisfying
\begin{equation}
2 \leq \inf_{\epsilon\in [\epsilon_1,\epsilon_2]} \abs{t_n(x_0 + i\epsilon)} \leq (2 K +1) M_n(\epsilon) + \mathcal{O}(\mathrm{e}^{ - 2 \pi q_n \delta_1}) ~\mbox{,}
\end{equation}
which implies the claim of the Lemma.
\end{proof}

Lemma \ref{lem_qac2} immediately strengthens  (\ref{eq_trrho7a})  in the sense:
\begin{equation}
\frac{1}{q_n} \int_{\abs{t_n(.+i\epsilon)}\geq 2} \log \left\vert t_n(x+i\epsilon) \right\vert \ud x= \frac{1}{q_n} \max\{\log M_n(\epsilon),0\} + o(1) ~\mbox{,}
\end{equation}
uniformly on any interval $[\epsilon_1,\epsilon_2]$ where $\min_{\epsilon \in [\epsilon_1,\epsilon_2]} L(\alpha, D_\epsilon) > 0$.

On the other hand considering the compact set $S:=\{ \epsilon \in [-\delta_1, \delta_1] :  L(\alpha, D_\epsilon) = 0\}$, it is automatically true that $\frac{1}{q_n} \log M_n(\epsilon) \to 0$ uniformly in $\epsilon \in S$ as $n\to \infty$. Thus, in summary we obtain the following asymptotic expression for the complexified LE under rational approximation of $\beta$:
\begin{equation} \label{eq_trrho7a2}
L(\alpha, D_\epsilon) = \frac{1}{q_n} \max\{\log M_n(\epsilon) , 0\} + o(1) = \max\{\max_{\abs{k} \leq K} \{\frac{1}{q_n} \log \abs{a_{n,k}} - 2 \pi \epsilon k \} ,0\} + o(1) ~\mbox{,}
\end{equation}
uniformly over $0\leq \abs{\epsilon} \leq \delta_1$ as $n\to \infty$. In the context of Schr\"odinger cocycles, we have thus established (\ref{eq_chambers1}) of Lemma \ref{lemma_avila}.

Equation (\ref{eq_trrho7a2}) shows that $L(r_{n}, D_{\epsilon})$ is uniformly close on $0 \leq \epsilon \leq \delta_1$ to a piecewise linear, convex function with right derivatives in
$\{ - 2 \pi K, \dots , 2 \pi K\}$. On the other hand as $r_{n} \to \alpha$, the continuity statement of \cite{C} for the Lyapunov exponent implies uniform convergence of $L(r_{n}, D_{\delta})$  to $L(\alpha, D)$, $\abs{\delta} \leq \epsilon$, which completes the proof.
\end{proof}

\section{Proof of Lemma \ref{lem_cheb}} \label{app_lemcheb}

Set
\begin{equation}
T(x) := \zeta T_n\left( \left(\dfrac{L}{\zeta}\right)^{1/n} \dfrac{x}{2^{1-1/n}} \right) ~\mbox{.}
\end{equation}
One checks that $T(x)$ shares the properties of $p(x)$ specified in Lemma  \ref{lem_cheb}. 

Consider appropriate horizontal shifts, so that $0$ is the leftmost point where both $\abs{T}$ and $\abs{p}$ equal $\zeta$. To simplify notation, we will still denote these shifted polynomials by $T$ and $p$, respectively. Consequently, let $N$ ($\widetilde{N}$) be the rightmost point where $\abs{T}$ ($\abs{p}$) attain $\zeta$. 

If $N \geq \widetilde{N}$, let $0=x_n < \dots < x_0=\widetilde{N}$ such that $p(x_k) = (-1)^k \zeta$, $0 \leq k \leq n$. In particular, since $\Vert T \Vert_{[0,\widetilde{N}]} \leq \zeta$, the definition of $x_k$ implies, $\abs{T(x_k)} \leq \abs{p(x_k)} = \zeta$.

Hence, considering $f:= p - T$ ($\deg f \leq n-1$), we conclude
\begin{equation}
\begin{cases}
f(x_k) \geq 0 & \mbox{, if $k$ even,} \\
f(x_k) \leq 0 & \mbox{, if $k$ odd.}
\end{cases}
\end{equation}
Since all $x_k$ are distinct, this requires $f$ to have at least $n$ zeros which, if $f \not \equiv 0$, however contradicts $\deg(f) \leq n-1$.

In case $N \leq \widetilde{N}$, interchange the roles of $T$ and $p$ in above proof. Thus, in summary we must have $p = T,$ as claimed.

\section{Proof of Lemma \ref{lem_shamissodin}} \label{app_shamissodin}

The proof of the Lemma is based on the following formula, which is only a slight alteration of Proposition 4.1 in \cite{D}. Proposition \ref{prop_interpolformula} is verified by a straightforward, albeit tedious computation:
\begin{prop} \label{prop_interpolformula}
Let $a_1, \dots, a_n \in \mathbb{R}$ and $s(x)$ be a differentiable function in a neighborhood of the points $x_j$ with $s(x_j) \neq 0$, $1 \leq j \leq n$. Consider,
\begin{eqnarray} \label{eq_definterpol}
\mathcal{T}(x;x_1, \dots, x_{n}) & := & \sum_{j=1}^{n} \dfrac{a_j}{s(x_j)} \prod_{ k \neq j} \dfrac{x - x_k}{x_j - x_k} + \prod_{j=1}^{n} (x-x_j) ~\mbox{,} \\
\mathcal{B}_j(x) & := & \prod_{ k \neq j} (x - x_k) ~\mbox{.}
\end{eqnarray}

Then, for any $x^*$ and $1 \leq k \leq n$ we have
\begin{equation}\label{b3}
\dfrac{\partial}{\partial x_k} \mathcal{T}(x^*; x_1, \dots, x_n) = - \dfrac{ \mathcal{B}_k(x^*)}{\mathcal{B}_k(x_k) s(x_k)} \dfrac{\partial}{\partial x} \vert_{x=x_k} \left[ \mathcal{T}(x; x_1, \dots, x_n) s(x) \right] ~\mbox{.}
\end{equation}
\end{prop}

We refer to the set-up and notation from the proof of Theorem \ref{thm_poly}. Comparing (\ref{eq_interpol}) with (\ref{eq_definterpol}) we conclude that the family of interpolating polynomials $\tau(x;x_1,\dots, x_{n-1})$, $x_1 < \dots < x_{n-1}$, with $s(x)= \mathrm{LC}(p) x$, satisfies the hypotheses of Proposition \ref{prop_interpolformula} with $a_j = (-1)^j \eta_j$.
Let $x^*$ be an arbitrary fixed point, $x^* \in B_j \setminus \{x_{j-1},x_j\}$. Then
\begin{eqnarray} \label{eq_xjs}
x_{n-1} < \dots < x_{j+1} < x_j < 0 < x^* < x_{j-1} < \dots < x_1 ~\mbox{, if $j$ is odd ~,} \\
x_{n-1} < \dots < x_{j+1} < x_j < x^* < 0 < x_{j-1} < \dots < x_1 ~\mbox{, if $j$ is even ~.} \nonumber
\end{eqnarray}

Let $1 \leq k \leq {n-1}$ be fixed and arbitrary. Based on Proposition \ref{prop_interpolformula}, we can estimate the sign of 
\begin{equation}\label{eq_b4}
\mathrm{sgn} \dfrac{\partial}{\partial x_k} \left \vert \tau(x^*;x_1,\dots,x_{n-1})\right \vert^2 = \mathrm{sgn}  (\tau(x^*;x_1,\dots,x_{n-1}) \dfrac{\partial}{\partial x_k} \tau(x^*;x_1,\dots,x_{n-1}) )~\mbox{.}
\end{equation}

Using that $\mathrm{LC}(p)\tau(x;x_1,\dots,x_{k-1}) s(x) \in \mathcal{F}$, $\mathrm{LC}(p) >0,$ and (\ref{eq_xjs}), we obtain (see Fig. \ref{figure_1}):
\begin{eqnarray} \label{eq_b5}
\mathrm{sgn} \left\{ \dfrac{\partial}{\partial x} 
 \left[ \tau(x;x_1,\dots,x_{n-1}) s(x) \right] \right\} & = & \begin{cases} (-1)^{k+1} & , ~ x=x_k^{+} ~\mbox{,} \\ (-1)^k &, ~x=x_k^- ~\mbox{.} \end{cases} \nonumber \\
\mathrm{sgn} \mathcal{B}_k(x^*) & = & \begin{cases} (-1)^{j-1} & , ~ k \geq j ~\mbox{and} ~ x^* \notin \{x_{j-1}, x_j\} ~\mbox{,} \\
                                                                                         (-1)^{j-2} & ,~k \leq j-1 ~\mbox{and} ~x^* \notin \{x_{j-1}, x_j\} ~\mbox{,} \\
                                                                                         0 &,~ x^* \in \{x_{j-1}, x_j\} ~\mbox{.} \end{cases} \nonumber \\
\mathrm{sgn} \mathcal{B}_k(x_k) & = & (-1)^{k-1} \nonumber \\
\mathrm{sgn} s(x_k) & = & \begin{cases} 1 &, ~k\leq j-1\mbox{,} \\ -1 &, ~ k \geq j \mbox{.} \end{cases}
\end{eqnarray}

Recall that, since we shifted the band $B_j$ so that $y_j=0$, $x^*$ is positive (i.e. $s(x^*) > 0$) if and only if $j$ is odd, and negative (i.e. $s(x^*) < 0$) otherwise (see (\ref{eq_xjs})). Thus in either case,
\begin{eqnarray} \label{eq_b6}
\mathrm{sgn} \tau(x^*) & = & \begin{cases} \mathrm{sgn}\left\{ \tau(x^*) s(x^*) \right\} = +1 &, ~ 0 < x^* ~\mbox{,} \\
                                                         - \mathrm{sgn}\left\{ \tau(x^*) s(x^*) \right\} = -1 &, ~ 0 > x^* ~\mbox{.} \end{cases} = (-1)^{j+1}
\end{eqnarray}

In summary, (\ref{b3}) and (\ref{eq_b4}) - (\ref{eq_b6}) imply
\begin{equation} \label{eq_deform}
\mathrm{sgn} \dfrac{\partial}{\partial x_k} \left \vert \tau(x^*;x_1,\dots,x_{n-1})\right \vert^2 = \begin{cases} -1 &, ~ x_k=x_k^- \mbox{,} \\ +1 &, ~x_k=x_k^+ \mbox{.} \end{cases}
\end{equation}

Finally, recall that for any $x$
\begin{equation}
\tau(x; x_1, \dots, x_{k-1}, x_k^+, x_{k+1}, \dots, x_{n-1}) = \tau(x; x_1, \dots, x_{k-1}, x_k^-, x_{k+1}, \dots, x_{n-1}) ~\mbox{,}
\end{equation}
whence any change of $x_k^-$ results in a corresponding change in $x_k^+$ and vice versa.  In fact, (\ref{eq_deform}) shows that increasing $x_k^-$ (thus decreasing $\left \vert \tau(x^*;x_1,\dots,x_{n-1})\right \vert$) results in a decrease of $x_k^+$. In particular, $\left \vert \tau(x^*;x_1,\dots,x_{n-1})\right \vert$ will be at minimum if and only if $x_k^- = x_k^+$. 

Since $1\leq k \leq n-1$ was arbitrary, we can deform $q(x)$ such that $x_k^- = x_k^+$ for all $k$, which yields the conclusion of Lemma \ref{lem_shamissodin} upon use of Lemma \ref{lem_cheb}.

\bibliographystyle{amsplain}

\end{document}